\documentclass[article]{jss}
\usepackage{multirow}
\usepackage{array}
\usepackage{color}
\newcolumntype{x}[1]{%
{\raggedleft\hspace{0pt}}p{#1}}%
\definecolor{gris5}{gray}{0.5}
\usepackage{float}
\usepackage[utf8]{inputenc}
\usepackage{amsmath,amsthm,amssymb}
\usepackage{xspace}
\newcommand{\R}{{\fontfamily{cmss}\fontseries{sbc}\fontshape{n}\selectfont R}\xspace}
\newcommand{\transp}{^{\text{\fontfamily{cmss}\fontseries{sbc}\fontshape{n}\selectfont T}}}
\newcommand{\vect}[1]{\lowercase{\boldsymbol{#1}}}
\newcommand{\mat}[1]{\boldsymbol{\mathcal{\uppercase{#1}}}}
\newcommand{\aleasca}[1]{\normalfont\text{\small \uppercase{#1}}}
\newcommand{\aleavect}[1]{\normalfont\textbf{\small\uppercase{#1}}}
\newcommand{\Log}{\textrm{Log}}
\newcommand{\Argmax}{\textrm{Argmax}}
\newtheorem{thm}{Theorem}

\makeatletter
\@nojsstrue
\makeatother

\author{C\'ecile Bordier\\GIN \And 
        Michel Dojat\\GIN \And
Pierre Lafaye de Micheaux\\Universit\'e de Montr\'eal}
\title{Temporal and Spatial Independent Component Analysis for fMRI data sets embedded in a \R package}

\Plainauthor{C\'ecile Bordier, Michel Dojat, Pierre Lafaye de Micheaux} %% comma-separated
\Plaintitle{Temporal and Spatial Independent Component Analysis for fMRI data sets embedded in a \R package} %% without formatting
\Shorttitle{Temporal and Spatial Independent Component Analysis for fMRI data sets} %% a short title (if necessary)

\Abstract{
For statistical analysis of functional Magnetic Resonance Imaging (fMRI) data sets, we propose a data-driven approach based on Independent Component Analysis (ICA) implemented in a new version of the \pkg{AnalyzeFMRI} \R package. For fMRI data sets, spatial dimension being much greater than temporal dimension, spatial ICA is the tractable approach generally proposed. However, for some neuroscientific applications, temporal independence of source signals can be assumed and temporal ICA becomes then an attracting exploratory technique. In this work, we use a classical linear algebra result ensuring the tractability of temporal ICA. We report several experiments on synthetic data and real MRI data sets that demonstrate the potential interest of our \R package.}
 
\Plainkeywords{keywords, comma-separated, not capitalized, Java} %% without formatting
\Address{
  C\'ecile Bordier and Michel Dojat\\
 INSERM U836 \\
  Universit\'e Joseph Fourier, Grenoble-Institut des Neurosciences (GIN)\\
  B\^atiment: Edmond J. Safra \\ Site Sant\'e 38706 La Tronche, France\\
  E-mail: \email{Michel.Dojat@ujf-grenoble.fr}\\
  URL: \url{http://nifm.ujf-grenoble.fr/~dojatm/}\\
$ $\\
  Pierre Lafaye de Micheaux\\
  Department of Mathematics and Statistics\\
  Universit\'e de Montr\'eal\\
  Montr\'eal, Qc, Canada\\
  E-mail: \email{lafaye@dms.umontreal.ca}\\
  URL: \url{http://www.biostatisticien.eu}
}
\graphicspath{{imgs/}}

\begin{document}

%% include your article here, just as usual
%% Note that you should use the \pkg{}, \proglang{} and \code{} commands.

%\section[About Java]{About \proglang{Java}}
%% Note: If there is markup in \(sub)section, then it has to be escape as above.

\section{Introduction}

Magnetic Resonance Imaging (MRI) is now a prominent non-invasive neuroimaging technique largely used in clinical routine and advanced brain research. Its success is largely due to a combination of at least three factors: 1) sensitivity of MR signal to various physiological parameters that characterize normal or pathological living tissues (such as diffusion properties of  $H_20$ molecules, relaxation time of proton magnetization  or blood oxygenation) leading to a vast panoply of MRI modalities (respectively restricted in our example to diffusion MR imaging, weighted structural images and functional MRI); 2) constant hardware improvements (e.g. mastering high field homogeneous magnets  and high linear magnetic field gradients respectively allows an increasing of spatial resolution or a reduction of acquisition time); and 3) sustained efforts in various laboratories to develop robust software: for image processing (to de-noise, segment, realign, fusion or visualize MR brain images), for computational anatomy leading to the exploration of brain structure modifications during learning, brain development or pathology evolution and for time course analysis of functional MRI data. Statisticians play a key role in this last factor since data produced are complex: noisy, highly variable between subjects, massive and, for functional data, highly correlated both spatially and temporally (\cite{Lange03}). \\ 
Functional MRI (fMRI) allows to detect the variations of cerebral blood oxygen level induced by the brain activity of a subject, lying inside a MRI scanner, in response to various sensory-motor or cognitive tasks (\cite{Chen99}). The fMRI signal is based on changes in magnetic susceptibility of the blood during brain activation. It is a non-invasive and indirect detection of brain activity: the signal detected is filtered by the hemodynamic response function (HRF) and the neuro-vascular coupling is  only partially explained (\cite{Logothetis04}). The main goal of fMRI experiments is to explore, in a reproducible way, the cortical networks implicated in pre-defined stimulation tasks in a cohort of normal or pathological subjects. The low signal to noise ratio obtained in functional images requires to repeat the sequence of stimuli several times (\cite{Henson04}) and to enroll a sufficient number of subjects (\cite{Thirion07}). In general, the data resulting from an fMRI experiment consist in a set indexed with time (typically many hundred) of 3D dimensional functional images with a $3\times3\times3~\text{mm}^3$ spatial resolution, and in a structural (or anatomical) image with a $1\times1\times1~\text{mm}^3$ resolution used to accurately localize functional activations. Note that a 3D image is in fact an array of many voxels's intensities. Various pre-processing steps are required to correct functional images from possible head subject movement, to realign functional and anatomical individual images and, for group studies, all individual data sets in a common referential. A spatial smoothing (e.g. using a gaussian kernel) is generally applied to functional images to compensate for potential mis-realignment and enhance the signal-to-noise ratio.  \\
Several frameworks have been proposed to date for statistical analysis of these pre-processed sets of functional data (see \cite{Lazar08}'s book for a recent review). The commonly used statistical approach, massively univariate, considers each voxel independently from each other using regression techniques (\cite{Friston95}); \cite{Bullmore97}). It is available in freeware packages such as FSL (\url{http://www.fmrib.ox.ac.uk/fsl/}), SPM (\url{http://www.fil.ion.ucl.ac.uk/spm/}), BrainVisa (\url{http://brainvisa.info/}) or NIPY (\url{http://nipy.sourceforge.net/}). 
The time series response at each voxel is modeled as a stationary linear filter where the finite impulse response corresponds to a model of the HRF. This leads to the specification of a general linear model (noted GLM thereafter; not to be confounded with the Generalized Linear Model) where the design matrix contains, for each time point, the occurrences of the successive stimuli (regressors) convolved with the HRF model. Other regressors can be seamlessly introduced  to model possible confounds. Many refinements to this approach have been proposed (\cite{Nichols02}; \cite{Friston05}; \cite{Roche07}). Spatial smoothness of the activated areas, normal distribution and independence of the error terms and a predefined form of the HRF used as a convolution kernel are the main  \textit{a priori} incorporated into the GLM. This model-driven approach allows to test, using standard Student or Fisher tests, the activated regions against a desired hypothesis by specifying compositions of regressors.  It is largely used essentially because of its flexibility in model specification allowing to test various hypothesis represented in corresponding statistical parametric maps. Clearly, the validity of the interpretation of these maps depends on the accuracy of the specified model. \\
An alternative exploratory (data-driven) approach relies on multivariate analysis based on Independent Component Analysis (ICA). ICA performs a blind separation of independent sources from a complex mixture of many sources of signal and noise. In this approach, relying on the intrinsic structure of the data, no assumptions about the form of the HRF or the possible  causes of responses are inserted. Only the number of sources or components to search  for could eventually be specified. To identify a number of unknown sources of signal, ICA assumes that these sources are mutually and statistically independent in space (sICA) or time (tICA). This assumption is particularly relevant to biological time-series (\cite{Friston98}). For fMRI data set analyses,  sICA is preferred because temporal points (few hundreds, corresponding to each occurrence of a functional image acquisition) are small compared to spatial ones (more than $10^5$, corresponding to the number of voxels contained in a functional image) leading for tICA to a computationnaly intractable mixing matrix (\cite{McKeown95}). However, temporal ICA could be relevant for some neuroscientific applications where temporal independence of sources can be assumed  (\cite{Calhoun01b}). In this context, these authors wrote \textit{``... Note that tICA is typically much more computationally
demanding than sICA for functional MRI applications because of a higher spatial than temporal dimension and can grow quickly beyond practical feasibility. Thus a covariance matrix on the order of $N^2$ (where $N$ is the number of spatial voxels of interest) must be calculated. A combination of increased hardware capacity
as well as more advanced methods for calculating and storing the covariance matrix may provide a solution in the future ..."}. In this paper, we propose to use a classical linear algebra result to alleviate the aforementioned computational burden. \\

The paper is structured as follows. First, in Section 2 we briefly describe the principle of temporal and spatial ICA in the context of fMRI data set analysis and detail the mathematical developments we propose for ensuring temporal ICA tractability.  In Section 3, we describe the current version of the \pkg{AnalyzeFMRI} \R package (see \cite{AnalyzeFMRI}), which is the first \R package designed for the processing and analysis of large anatomical and functional MRI data sets. It was initiated by J. Marchini (\cite{Marchini02}), who passed the torch in 2007 to the third author of this paper. This package includes, compared to its initial version, our recent extensions: i.e. NIFTI format management, cross-platform visualization based on Tcl-Tk components and temporal (and spatial) ICA (TS-ICA). We report, in Section 4, results using synthetic data and  real MRI data sets coming from human visual experiments, obtained using TS-ICA.  Finally, we conclude about the interest of the \pkg{AnalyzeFMRI} package and our extensions for the exploration of MRI
data and outline our plans for future extensions.

\section{Spatial and Temporal Independent Component Analysis}
Independent component analysis (ICA) is a statistical technique whose aim is to recover hidden underlying source signals from an observed mixture of these sources. In standard ICA, the mixture is supposed to be linear and the only hypothesis made to solve this problem (known as the blind source separation problem) is that the sources are statistically mutually independent and are not Gaussian. \\

The generative linear instantaneous noise-free mixing ICA model is generally written under the form
\begin{equation}\label{model1}
\aleavect{x}=\mat{A}\aleavect{s}
\end{equation}
where $\aleavect{x}=(\aleasca{x}_1,\ldots,\aleasca{x}_m)\transp$ is the $m\times1$ continuous-valued random vector of the observable  signals, $\mat{A}=(a_{ij})$ is the unknown constant (non random) and invertible square mixing matrix of size $m\times m$ and $\aleavect{s}=(\aleasca{s}_1,\ldots,\aleasca{s}_m)\transp$ is the $m\times1$ continuous-valued random vector of the $m$ unknown source signals to be recovered. Note that if we denote by $\mat{B}$ the inverse of matrix $\mat{A}$, then we can write $\aleavect{s}=\mat{B}\aleavect{x}$. The term ``recover'' here means that we want to be able, based on an observed sample $\vect{x}_1,\ldots,\vect{x}_n$ (possibly organised in a matrix $\mat{x}$ of size $n\times m$) of the random vector $\aleavect{x}$, to estimate the densities $f_{\aleasca{s}_j}$ of the $m$ sources $\aleasca{s}_j$, or at least to be able to build an ``observed'' sample of size $n$ of each one of these $m$ sources, which are usually called the independent (extracted) components. For example, this sample could be computed, if one has an estimate $\hat{\mat{B}}$ of the separating matrix $\mat{B}$, as $\vect{s}_1,\ldots,\vect{s}_n$ where $\vect{s}_i=\hat{\mat{B}}\vect{x}_i$, $1\leq i\leq n$.\\

Note also that, using the independence property of the sources, the density of the random vector $\aleavect{x}$ can be expressed as
\[
 f_{\aleavect{x}}(\vect{x})=|\mat{A}^{-1}|f_{\aleavect{s}}(\vect{s})=|\mat{B}|\prod_{j=1}^m f_{\aleasca{s}_j}(s_j).
\]

It then follows that one can write the Log-likelihood of the observed sample as
\[
\Log~\mathcal{L}(\mat{B})=\Log \prod_{i=1}^nf_{\aleavect{x}}(\vect{x}_i)=n\Log|\mat{B}|+\sum_{i=1}^n\sum_{j=1}^m \Log~f_{\aleasca{s}_j}(\vect{b}_j\transp\vect{x}_i)
\]
where $\vect{b}_j$ denotes the $j^{th}$ column of $\mat{B}$. This is easy to prove when one notices that $\aleasca{s}_j=\vect{b}_j\transp\aleavect{x}$.\\
Now, it remains to compute $\hat{\mat{B}}=\Argmax~\Log~\mathcal{L}(\mat{B})$ using some optimization algorithm. To perform this operation, prior densities for the sources, or a simple parametrization of the sources can be considered (see details in \cite[p.205-6]{Hyvarinen01}). Alternatives approaches, not necessarily  based on the likelihood function, are available to estimate $\mat{B}$ (and thus $\aleavect{S}$). For example, there is a relation between independence and non gaussianity (\cite{Cardoso03}). In our package, we used the FastICA algorithm which consists in finding the sources that are maximally non Gaussian, where non gaussianity is measured using the kurtosis, see \cite{Hyvarinen01}.\\

To apply standard ICA techniques on fMRI data sets, the first step is to obtain a 2D data matrix $\mat{X}$ from the 4D data array resulting from an fMRI experiment (the 4D array is the concatenation in time of several 3D functional volumes). 
This can be performed in two (dual) ways:
\begin{itemize}
 \item[(a)] one may consider that the data consist in the realization of $t_l$ random variables, each one measured (sampled) on $v_l$ voxels. This results in $t_l$ 3D spatial maps of activation. Each 3D map is then unrolled (in an arbitrary order) to get a matrix $\mat{X}$ of size $v_l\times t_l$. The mixing matrix $\mat{A}$ is in this case of size $t_l\times t_l$. 
 \item[(b)] one may consider that the data consist in the realization of $v_l$ random variables, each one measured at $t_l$ time points. This results in $v_l$ time courses each one of length $t_l$, collected into a matrix $\mat{X}$ of size $t_l\times v_l$ (here again, the order of the $v_l$ time courses in the resulting matrix is arbitrary). The mixing matrix $\mat{A}$ is in this case of size $v_l\times v_l$.
\end{itemize}
Using these data, the empirical counterpart of the noise-free model (\ref{model1}) can then be written as
\begin{equation}\label{model2}
 \mat{X}\transp=\mat{A}\mat{S}\transp
\end{equation}
where $\mat{X}=\left[\begin{array}{c}
\vect{x}_1\\
\vdots\\
\vect{x}_n
\end{array}\right]$ and $\mat{S}=\left[\begin{array}{c}
\vect{s}_1\\
\vdots\\
\vect{s}_n
\end{array}\right]$.\\

Case (a) corresponds to spatial ICA (sICA) and the rows of matrix $\mat{S}\transp$ contain spatially independent source signals of length $n=v_l$ (unrolled source spatial maps). Case (b) corresponds to temporal ICA (tICA) and the rows of matrix $\mat{S}\transp$ contain here temporally independent source signals of length $n=t_l$ (source time courses). Note that the row-dimension of  matrices $\mat{X}$ and $\mat{S}$ above corresponds to sample size, which is the classical statistical community's convention (but not  the neuroimaging community one where matrices should be transposed).\\

At this point, one may have noticed that, because the mixing matrix $\mat{A}$ is square in standard ICA \cite[p.267]{Hyvarinen01}, in writing (\ref{model2}) we have implicitly supposed that the number of sources $m$ is equal to $t_l$ in case (a) and $v_l$ in case (b). This is not necessarily the case. Data pre-processing based on PCA is generally used to overcome this problem. Doing this, model (\ref{model2}) should be re-written as
\begin{equation}\label{model3}
\boldsymbol{\Lambda}_{red}^{-1/2}\mat{E}\transp_{red}\dot{\mat{X}}\transp=\mat{A}\mat{S}\transp
\end{equation}
where $\boldsymbol{\Lambda}_{red}$ (resp. $\mat{E}_{red}$) is the (\textit{reduced}) matrix whose diagonal elements (resp. columns) consist of the $m$ largest (non null) eigenvalues (resp. eigenvectors) of the empirical covariance (or eventually correlation) matrix $\dot{\mat{X}}\transp\dot{\mat{X}}/n$  (note that the mixing matrix into $\dot{\mat{X}}\transp$ is then given by $\mat{a}^X=\mat{E}_{red}\boldsymbol{\Lambda}^{1/2}_{red} \mat{A}$). The size of the matrix $\mat{E}_{red}$ is respectively $t_l\times m$ in case (a) and $v_l\times m$ in case (b). Note also that the Singular Value Decomposition (SVD) of matrix $\dot{\mat{x}}$ can be used:
\[
 \dot{\mat{x}}=\mat{u}\mat{d}\mat{v}\transp
\]
and then replace, in equation (\ref{model3}), $\boldsymbol{\Lambda}_{red}$ with $\mat{d}^2_{red}/n$ where $\mat{d}_{red}$ is the diagonal matrix consisting of the $m$ largest singular values of $\mat{d}$, and $\mat{E}_{red}$ with $\mat{v}_{red}$ refering to the associated singular vectors. Equation (\ref{model3}) then leads to the following decomposition:
\begin{equation}\label{decompX}
 \dot{\mat{x}}\transp=\frac{1}{\sqrt{n}}\mat{v}_{red}\mat{d}_{red}\mat{a}\mat{s}\transp=\sum_{j=1}^m\mat{a}^X_{\bullet j}\otimes\mat{s}_{\bullet j},
\end{equation}
where $\mat{s}_{\bullet j}$ denotes the $j^{th}$ column of $\mat{S}$. Note that the pair $(\mat{a}^X_{\bullet j},\mat{s}_{\bullet j})$ is sometimes (abusively) called the $j^{th}$ independent (estimated) component, although this term should be used solely for $\mat{s}_{\bullet j}$, whereas $\mat{a}^X_{\bullet j}$ refers to the weighting coefficients (degree of expression) of the $j^{th}$ spatial component over time (for sICA) or of the $j^{th}$ temporal source over space, i.e. over the voxels (for tICA).\\

Figure (\ref{decompsICA}) below is an illustration of equation (\ref{decompX}) for sICA.

\begin{figure}[H]
\centering
 \includegraphics{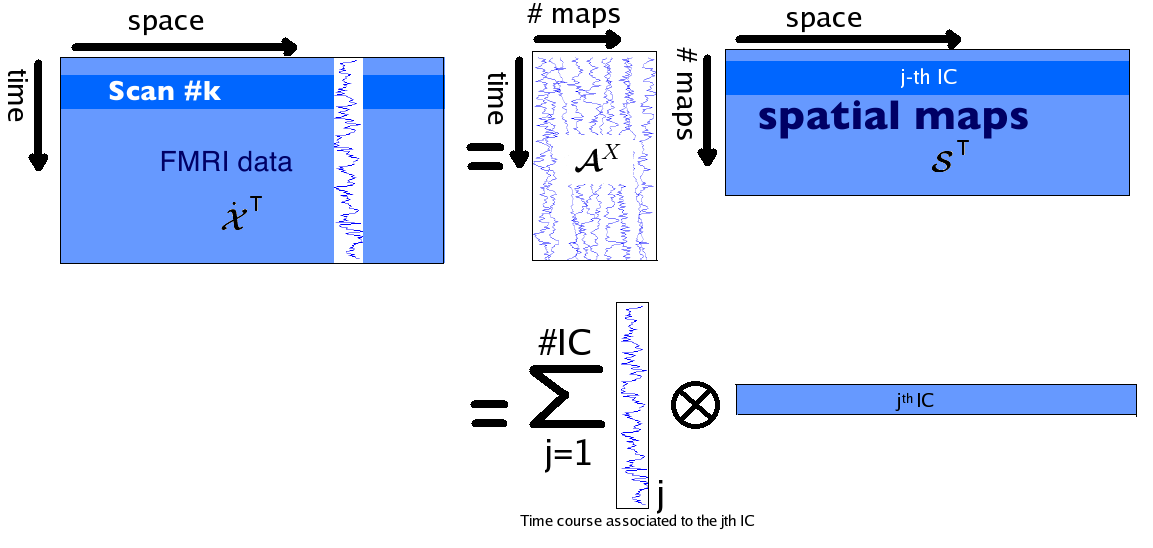}
\caption{Illustration of the sICA decomposition after a PCA pre-processing step.}\label{decompsICA}
\end{figure}

Now, due to the large number of voxels in fMRI experiments, it is not computationally tractable to fully diagonalize the correlation matrix in the temporal case (which is in this case of size $v_l\times v_l$). So tICA, as far as we know, has never been applied on the entire brain volume but only on a small portion of it (\cite{Calhoun01b}; \cite{Seifritz02};  \cite{Hu05}). \\

Our extension to the  \R package \pkg{AnalyzeFMRI} for TS-ICA uses a nice property of the SVD decomposition that allows to obtain the non-zero eigenvalues (and their associated eigenvectors) of the correlation matrix in the temporal case. It then becomes feasible to perform tICA for fMRI data on the whole brain volume. We now briefly present this result.

\begin{thm}
The largest eigenvalues of the (huge) covariance matrix in the temporal case, as well as their associated eigenvectors, can be obtained from the same quantities  computed from the (small) covariance matrix in the spatial case.
\end{thm}
\begin{proof}
We consider the temporal case, where the size of the matrix $\dot{\mat{x}}$ is $t_l\times v_l$. Let's note $\mat{S}_X=\dot{\mat{x}}\transp\dot{\mat{x}}/t_l$ the (empirical) covariance matrix of $\dot{\mat{x}}$, which (large) size $v_l\times v_l$. We want to find the $r$ nonzero largest eigenvalues of $\mat{S}_X$ and their associated eigenvectors $\vect{f}_k$, $k=1,\ldots,r$. SVD theory allows to write
\begin{equation}\label{equ1}
 \dot{\mat{x}}\transp\dot{\mat{x}}\vect{f}_k=d_k^2\vect{f}_k \qquad k=1,2,\ldots,r;
\end{equation}
\begin{equation}\label{equ2}
 \dot{\mat{x}}\dot{\mat{x}}\transp\vect{g}_k=d_k^2\vect{g}_k \qquad k=1,2,\ldots,r.
\end{equation}
Pre-multiplying equation (\ref{equ2}) by $\dot{\mat{x}}\transp$, one can see that $\dot{\mat{x}}\transp\vect{g}_k$ is an eigenvector of $\dot{\mat{x}}\transp\dot{\mat{x}}$ associated with the eigenvalue $d_k^2$. Thus, $\vect{f}_k$ is proportional to $\dot{\mat{x}}\transp\vect{g}_k$.\\ %Similarly, $\vect{F}_k$ is proportionnal to $\dot{\mat{x}}\transp\vect{G}_k$.\\

The idea is thus to compute the $t_l$ eigenvalues $\{d_1^2,\ldots,d_{t_l}^2\}$ and the $t_l$ eigenvectors $\vect{g}_k$ of the (small) matrix  $\dot{\mat{x}}\dot{\mat{x}}\transp$ of size $t_l\times t_l$. From this point, we get the $t_l$ first eigenvectors $\vect{f}_k$ (among the $v_l$ ones) of $\mat{s}_X$ using this formula:
\begin{equation}
 \vect{f}_k=\frac{1}{d_k}\dot{\mat{x}}\transp\vect{g}_k.
\end{equation}
% which can be put upon matrix form as
% \begin{equation}
%  \mat{f}=\dot{\mat{x}}\transp\mat{g}\mat{d}^{-1}.
% \end{equation}

The $v_l$ eigenvalues of $\mat{s}_X$ are given by $\frac{1}{t_l} d_1^2,\ldots,\frac{1}{t_l} d_{t_l}^2,0,\ldots,0$. Note that the last $v_l-t_l$ eigenvectors of $\mat{s}_X$ cannot be obtained using this approach, but anyway, as $d_i^2=0$ ($i>t_l$) they  do not contain any useful information.\\

\end{proof}

\section{The AnalyzeFMRI package}
\pkg{AnalyzeFMRI} is a package for the exploration and analysis of large 3D MR structural data sets and 3D or 4D MR functional data sets. From reconstructed MR volumes, this package allows the user to examine data quality and analyze time series. To efficiently explore fMRI data sets using tICA and sICA we added several interesting extensions to the initial package (e.g. tICA, automatic choice of the number of components to extract or GUI visualization tool). Some of them  are briefly described below  (see \url{http://user2010.org//tutorials/Whitcher.html} for more details, and also \cite{Marchini02} for a description of initial functions). Table \ref{table1} describes seven important functions available in the package.\\

\textit{Importing data:}\\
The package now provides read and write capabilities for the new NIFTI  (\texttt{nii} or \texttt{hdr/img} files) format. This format contains a header gathering all the volume information (image dimension, voxel dimension, data type, orientation, quaternions, ..., up to more than 40 parameters) and a data part that contains values corresponding to the MR signal intensity measured at each voxel of the image object. \\

\textit{Data pre-processing:}\\
Briefly, before doing any statistical analysis, functional MR data should be corrected from geometric distortions, realigned and smoothed. Only the latter step is embedded into the current (and initial) version of the package.\\

\textit{Image operators:}\\
Several operators can be applied on the images such as rotation, translation, scaling, shearing or cropping. These operations can be performed by changing quaternion parameter in the NIFTI header or by direct modification of the matrix values. The matrix indices (voxel position) can be translated to volume coordinates (in mm) to facilitate comparison between subjects.\\

\textit{Data analysis using TS-ICA:}\\
In the initial version of the package, it was only possible to analyze fMRI data using spatial ICA. We added temporal ICA and the automatic detection of the number of components to extract. Automatic detection is based of the computation of the eigenvalues of the empirical correlation matrix of the data, keeping only those greater than 1. The automatic detection is useful when no \textit{a priori}  knowledge is available. Note also that the user can now insert \textit{a priori} knowledge in selecting only a specific region of the brain to explore (\textit{via} a mask image) or in searching for components correlated with a specific time course signal.\\

\textit{Visualization:}\\
Anatomical or functional volumes and statistical (parametric or not) maps can be displayed in two separate windows with linked cursors to localize a specific position  (see Figure \ref{visu}).  Our visualization tool can be used in two ways. First, you can use it to visualize the results of a temporal or spatial ICA (as displayed in Figure \ref{visu} for sICA). The time slider here indicates the rank of the component currently visualized (among all those extracted) and the displayed time course represents the values of the spatial component for the selected voxel (blue circle). Second, you can use it to visualize raw fMRI data. In this case the time slider would represent the time course of the selected voxel, i.e. the MR signal values across time measured at the voxel position.

\begin{figure}[h!]
\centering
\includegraphics[width=\textwidth ]{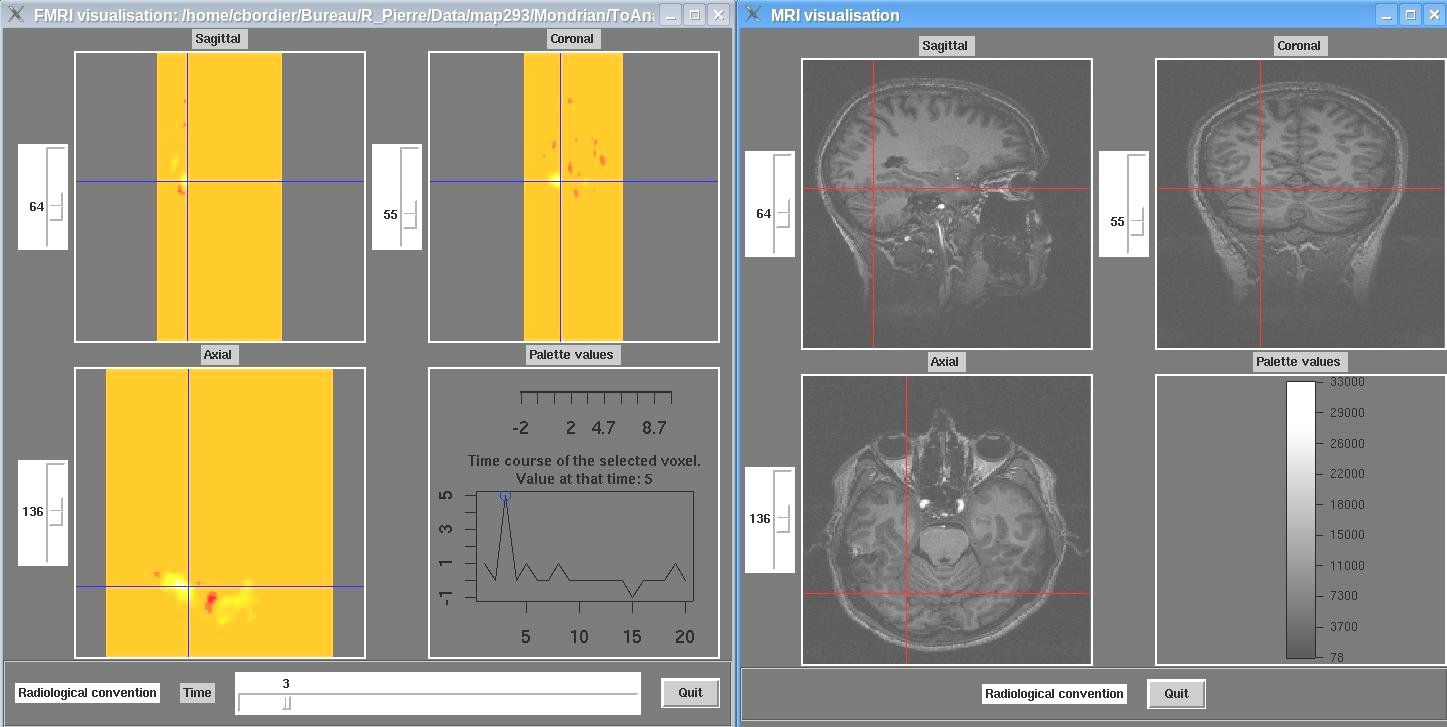}
\includegraphics[width=\textwidth,height=2.5cm ]{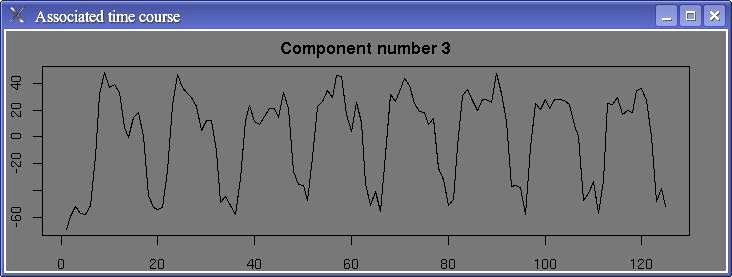}
\caption{Image Display. Right top: Anatomical image (clockwise: sagittal, coronal and axial views). Left top:  statistical map of activations obtained after spatial IC analysis of the functional data sets in the sagittal, coronal and axial orientation.  The value of the selected extracted spatial component (here rank=3) for the selected voxel (blue cross) is indicated in the right bottom quadrant (blue circle). The localization of the selected voxel is reported on the anatomical image (red cross). Bottom: Time course of the weighting coefficients of the third component (identical for all the voxels of this component).}
\label{visu}
\end{figure}

\begin{center}
\begin{table}
\centering
\begin{tabular}{| l |p{8cm}|}
\hline
\multicolumn{1}{|c|}{\R function} & \multicolumn{1}{c|}{Description}\\
\hline
\code{f.analyzeFMRI.gui()} & Starts an \R/TclTk based GUI to explore, using the \pkg{AnalyzeFMRI} package functions, an fMRI data set stored in ANALYZE format. \tabularnewline
\hline
\code{f.icast.fmri.gui()} & The GUI provides a quick and easy to use interface for applying spatial or temporal ICA to fMRI data sets in NIFTI  format.\tabularnewline
\hline
\code{f.plot.volume.gui()} & TclTk GUI to display functional or structural MR images.  This GUI is useful for instance to display the results performed with \code{f.icast.fmri.gui()}.\tabularnewline   
\hline
\code{f.read.header(file)} & Reads ANALYZE or NIFTI (\texttt{.hdr} or \texttt{.nii}) header file. The format type is automatically detected by first reading the  magic field. \tabularnewline
\hline
\code{f.read.volume(file)} & Reads ANALYZE or NIFTI image file and puts it into an array. Automatic detection of the format type. \tabularnewline
\hline
 \code{f.write.analyze(mat,file,...,)} & Stores the data in ANALYZE format: creation of the corresponding  \texttt{.img/.hdr} pair of files.
 \tabularnewline   
\hline
 \code{f.write.nifti(mat,file,size,...)}   & Stores the data in NIFTI format: creation of the corresponding  \texttt{.img/.hdr} pair of files or single \texttt{.nii} file.
  \tabularnewline   
  \hline 
\end{tabular}
\caption{Seven main functions of our package with their description.\label{table1}}
\end{table}
\end{center}

%\begin{figure}
%\begin{center}
%\includegraphics[width=3in, height=3in]{fig1.png}
%\caption{ceci est un exemple}
%\end{center}
%\end{figure}

\section{Results}
We evaluated the TS-ICA part of the \pkg{AnalyzeFMRI} package both on simulated data and real data sets coming from human visual fMRI experiments.
\subsection{Simulated data sets }
In fMRI experiments, three standard paradigms are used. ``Block design" which alternates, in a fixed order, stimuli that last few seconds;  ``event-related design" which alternates, in a random or pseudo-random order, stimuli that last few milliseconds and ``phase-encoded paradigm" that generates traveling periodic waves of activation with different phases.  
In order to detect patterns of activation for the two former cases, we can use respectively a cross correlation with a square wave, or a binary cross correlation (to be defined later) with a sequence of 0 and $\pm1$ representing the stimulation conditions. A Fourier analysis is more suitable for the latter.

Before testing our method on real data sets, we used three simulated cases: 1) a simple case to show how works our method, and two cases simulating real conditions: 2) an event-related design simulation and 3) a phase-encoded simulation. The latter  simulates the real case described in Section 4.2 ``retinotopic mapping experiment". The square wave signal in the former simulates the ``color center experiment" reported in Section 4.2.\\

\R source code (including comments) for each one of the three aforementioned simulations is provided as supplementary material. 
Because our final results may change due to the use of random numbers (simulated data and initial conditions for ICA algorithm), 
we provided, in our \R code, the seeds we used for the random generators. This will permits the reader to obtain exactly the same results as those presented here.

\subsubsection{Various independent sources simulation}
The first simulated data set consisted in a sequence comprising 100 3D-images. Each image ($128\times128\times3$ voxels) was composed of four \textbf{partially overlapping} and concentric tubes. Each tube contained a single signal in its non overlapping part and a sum of two signals in its parts that intersect with another tube.  For single signals, we used, from the central tube to the peripheral one respectively, a sinusoid ($f=1/11~Hz$, $\phi=0$), a square wave  ($f=1/10~Hz$, $\phi=0$), a sinusoid ($f=1/16~Hz$, $\phi=0$) and a square wave ($f=1/4~Hz$, $\phi=0$). The background, which overlaps the tube at the periphery, contained, in its non overlapping part, a Gaussian noise (sd=0.2) (see Figure \ref{circle_simul1}).  Thus, four pure signals and four mixed signals were considered. To be realistic, we also added everywhere a Gaussian noise (sd=0.1). 
\begin {figure}[h!]
\centering
\begin{tabular}{cc}
\multirow{2}{*}{
\includegraphics[width=0.25\textwidth ]{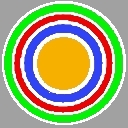}}\\%
&\includegraphics[width=0.4 \textwidth]{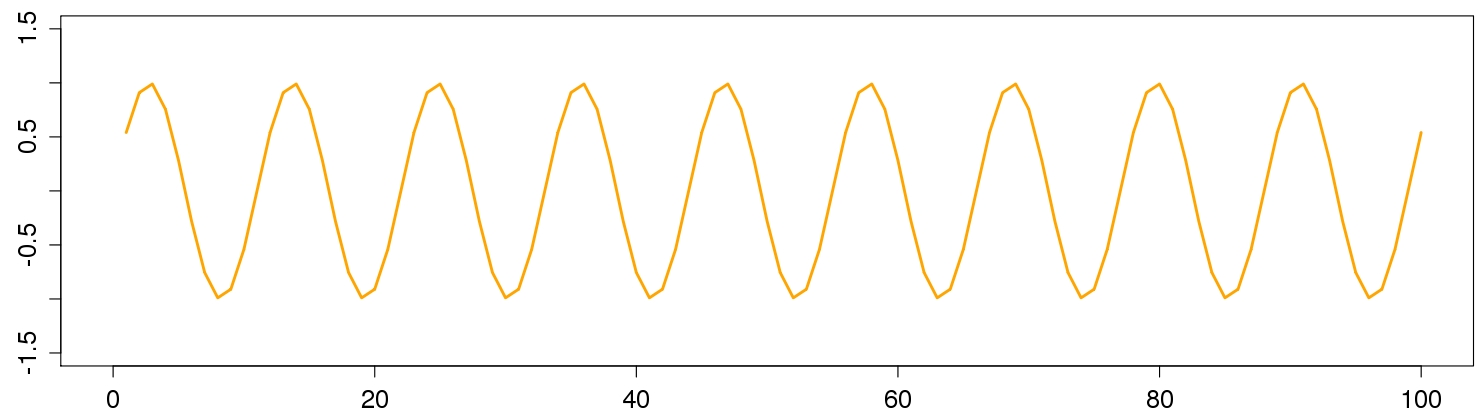}\\
& \includegraphics[width=0.4 \textwidth]{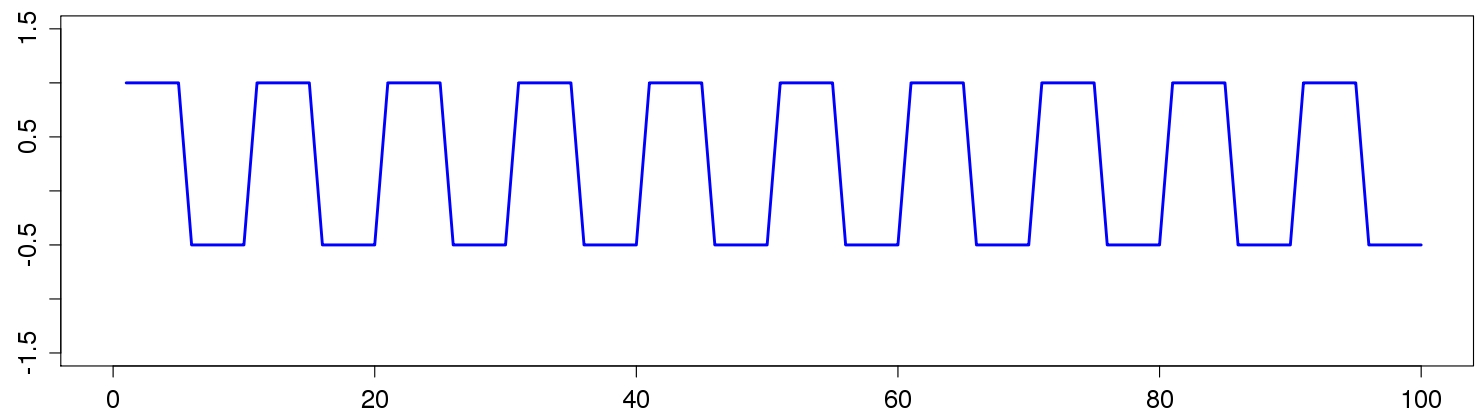}\\
 \includegraphics[width=0.4 \textwidth]{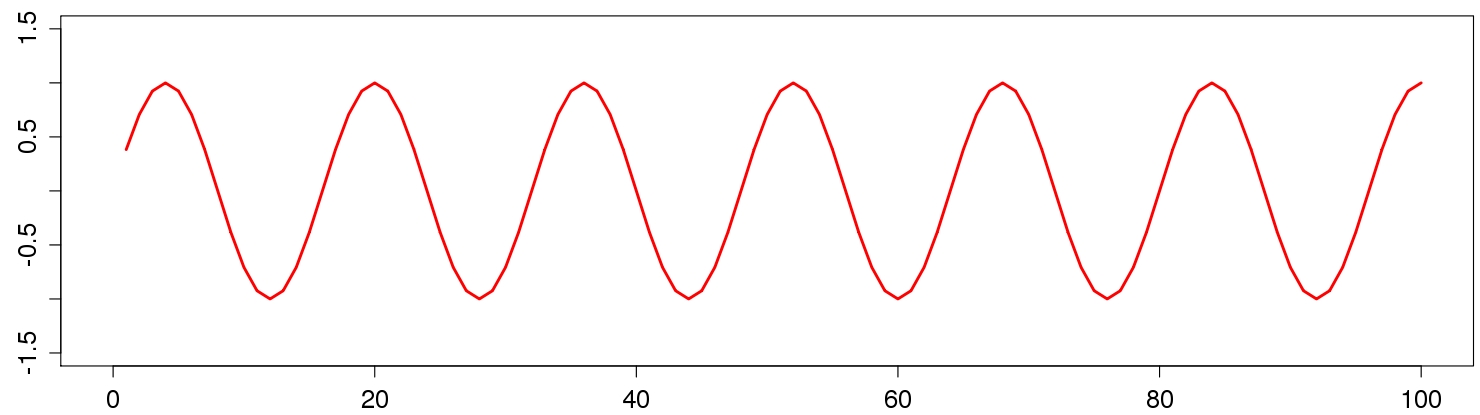}
& \includegraphics[width=0.4 \textwidth]{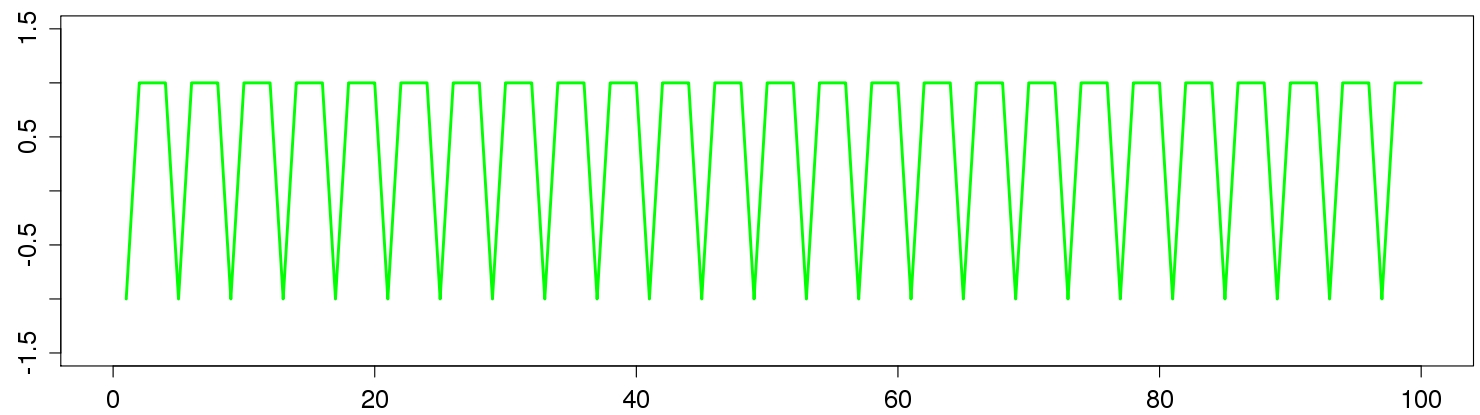}\\
\end{tabular}
\caption{First simulated data set. Upper left: A transverse slice of the volume. Each color indicates the localization of each signal. Pure signals are represented in orange (source 1), blue (source 2), red (source 3) and green (source 4), and mixed signals are present in white parts. Background (grey) contains a Gaussian noise (sd=0.2). Time courses of each pure signal are displayed with their corresponding color.}
\label{circle_simul1}
\end{figure}

We applied temporal and spatial ICA to these simulated data. Figure \ref{circle_simul1_result} shows the time course of the different extracted components and their spatial localization. It is interesting to note that temporal ICA extracted automatically four components with relevant time course and localization that appears correct using our \R function \code{f.plot.volume.gui()}. The computed frequencies of the time course of these components were respectively, when ordered from center to periphery, nearly equal to 1/11~Hz, 1/10~Hz, 1/16~Hz and 1/4~Hz with phase difference $\approx0$ (modulo $\pi$) with the corresponding original source signal. We used \R functions \code{Mod(fft(signal))} and \code{Arg(fft(signal))} to compute these quantities. Spatial ICA extracted automatically, in the non overlapping parts,  four spatial components with form and localization approximatively comparable to the initial sources. The first one (central tube) was not extracted. Note that each extracted component was associated with one of the original sources. This  association was made based on the higher absolute value of the correlation between the time course of the component and each one of the four original signals. Thresholded localization of a specific component was then computed by keeping its voxels with values higher (resp. lower) than their empirical quantile of order 0.9 (resp. 0.1) if the correlation of its  time course with the associated original signal was positive (resp. negative). The frequencies of the time course of the extacted components 1 to 4 were found to be, respectively, nearly equal to 1/4~Hz, 1/10~Hz, 1/16~Hz and 1/16~Hz, with phase difference $\approx0$ (modulo $\pi$) with the corresponding original source signal. The localizations were less accurate than the ones obtained with temporal ICA. This is not surprinsing. Indeed, a nonparametric test for the mutual independence between our source time signals was performed using the \R package \pkg{IndependenceTests} (for more information see \cite{IndependenceTests}, \cite{Beran07} or \cite{Bilodeau05}). 

\begin{verbatim}
require(IndependenceTests)
dependogram(cbind(signal1,signal2,signal3,signal4),c(1,1,1,1),N=10,B=200) 
\end{verbatim}

\begin {figure}[h!]
\centering
\includegraphics[width=6cm]{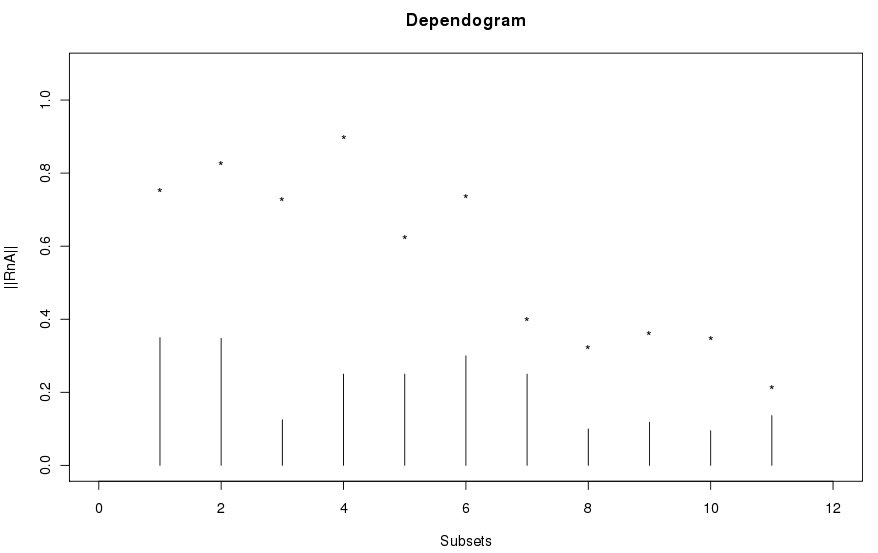}
\caption{Test of the mutual independence between our four original signals.}
\label{indep}
\end {figure}

It was not possible to detect any form of dependence among these four source signals (see Figure \ref{indep}). On the other hand, the spatial sources were not (spatially) independent because of their overlapping parts, and indeed only portions with pure signal were correctly extracted using sICA.
\begin {figure}[h!]
\centering
\begin{tabular}{m{0.4\textwidth}m{0.05\textwidth}m{0.05\textwidth}m{0.4\textwidth}}
\includegraphics[width=0.4\textwidth ]{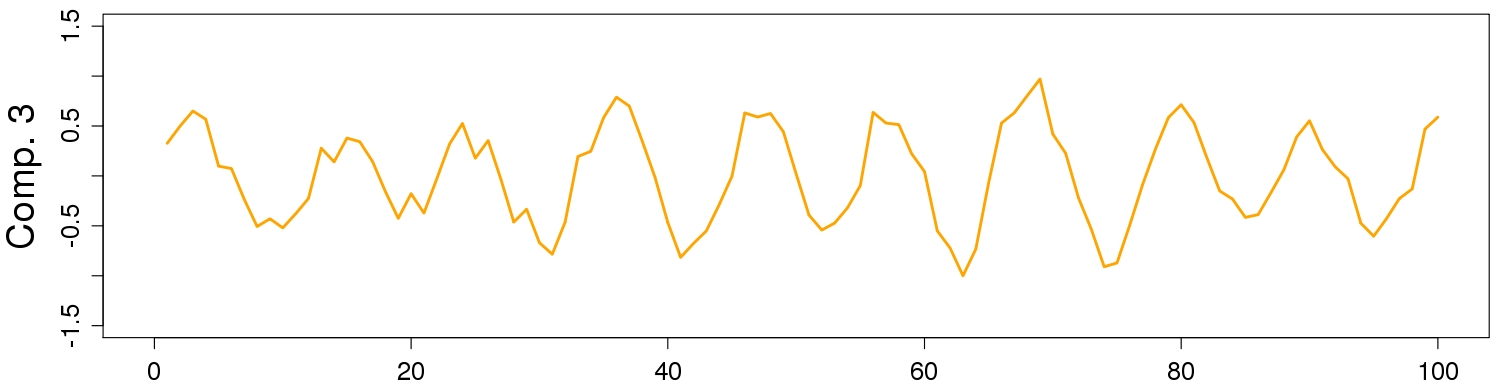}&
\includegraphics[width=0.05\textwidth ]{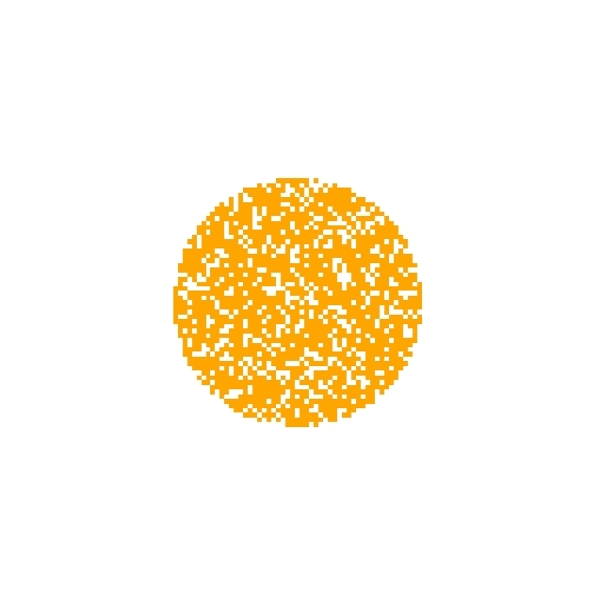}&
\includegraphics[width=0.05\textwidth ]{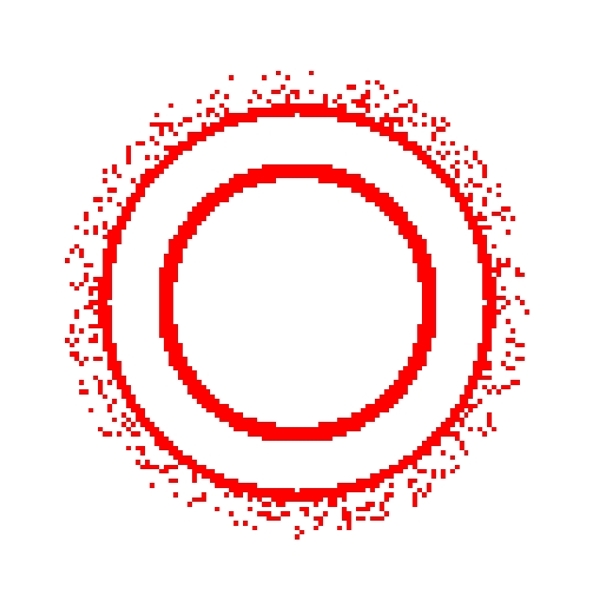}&
\includegraphics[width=0.4\textwidth ]{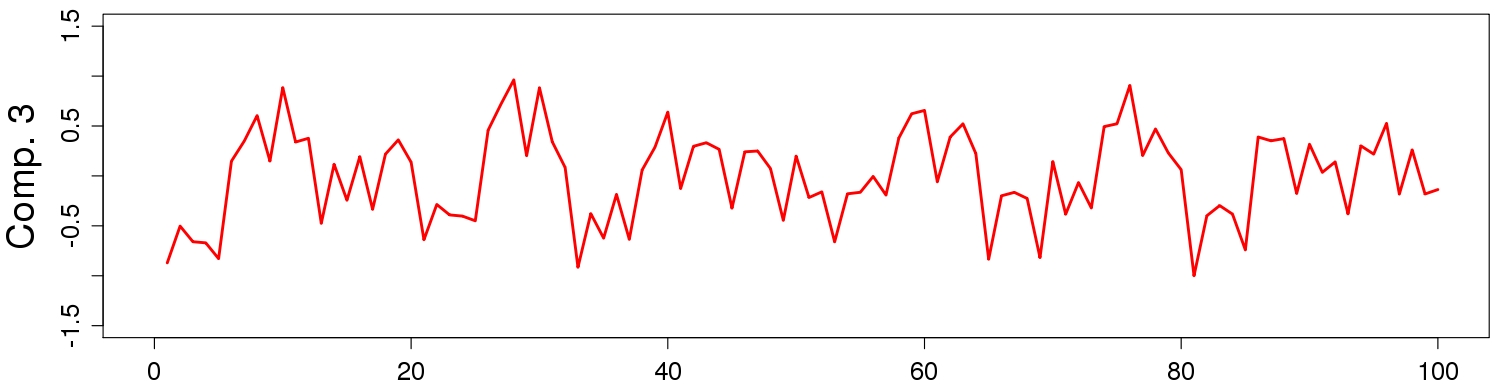}\\%
\includegraphics[width=0.4\textwidth ]{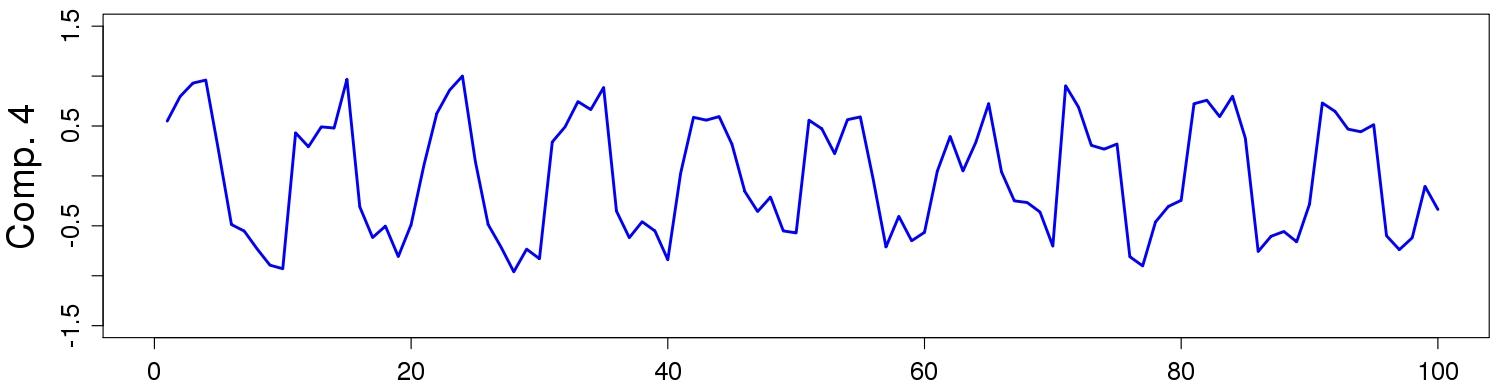}&
\includegraphics[width=0.05\textwidth ]{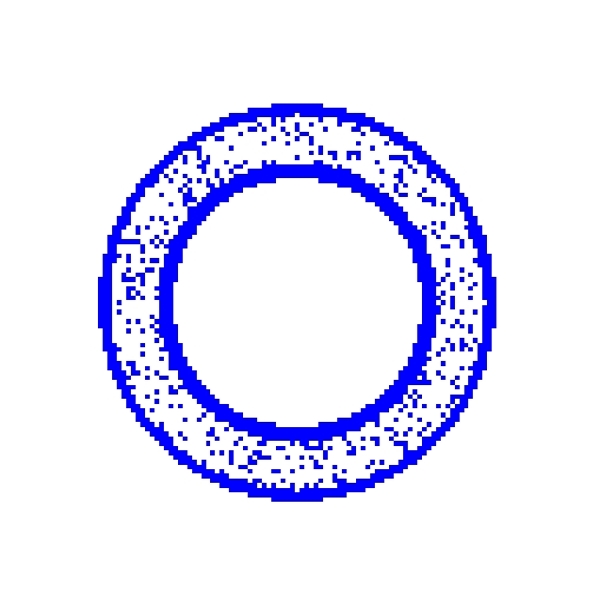}&
\includegraphics[width=0.05\textwidth ]{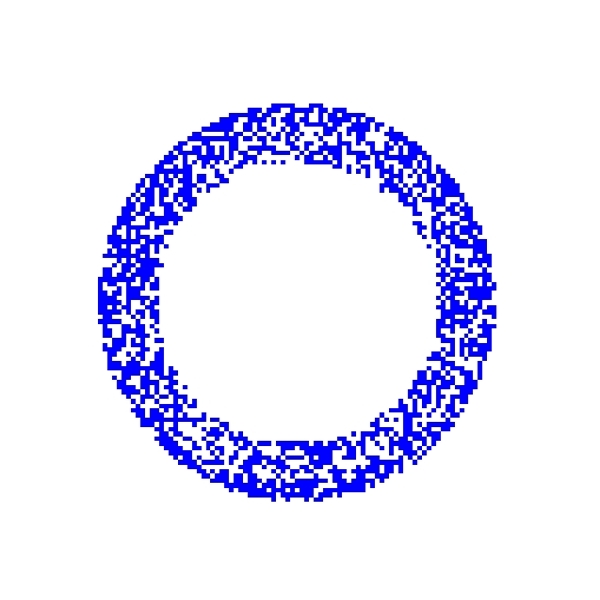}&
\includegraphics[width=0.4\textwidth ]{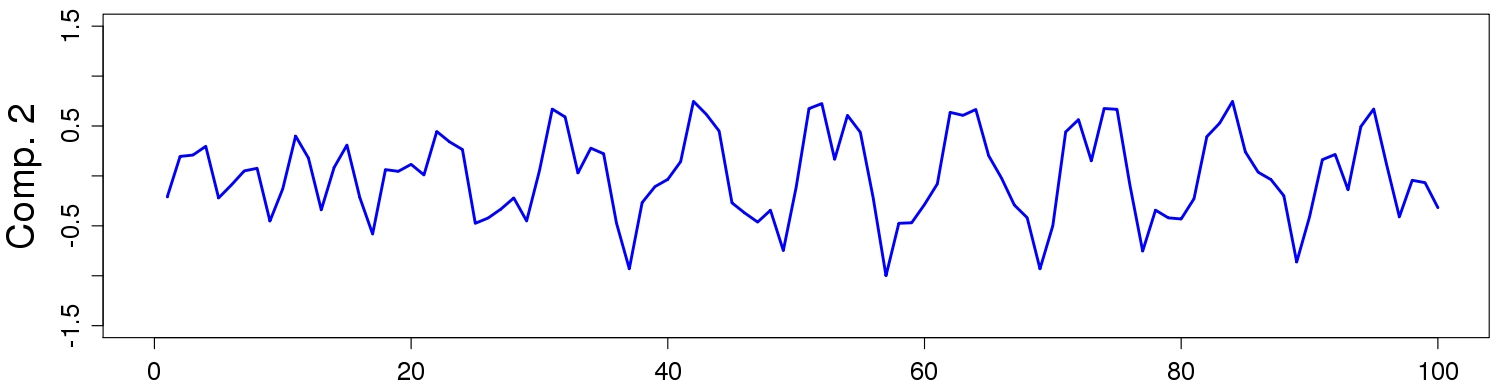}\\%
\includegraphics[width=0.4\textwidth ]{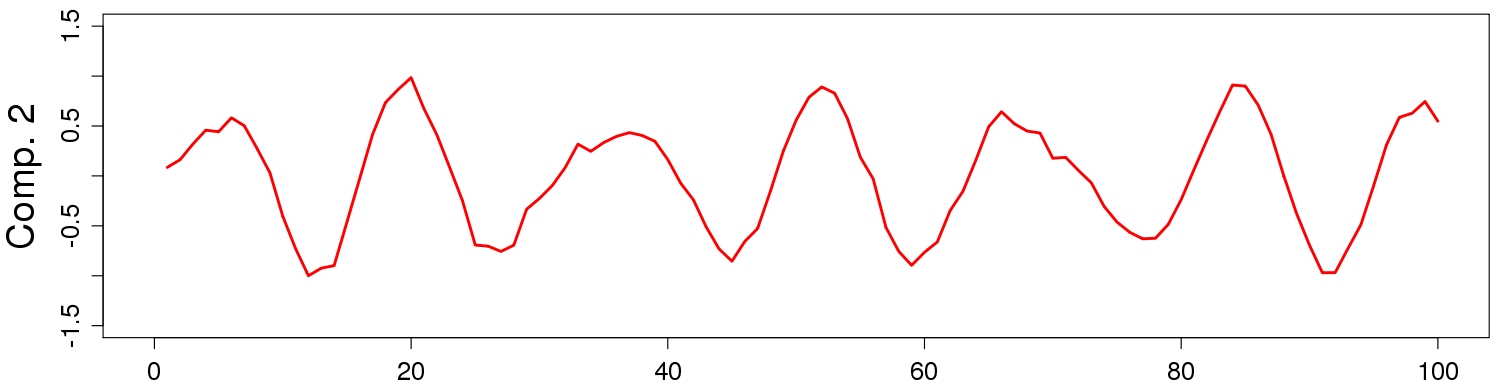}&
\includegraphics[width=0.05\textwidth ]{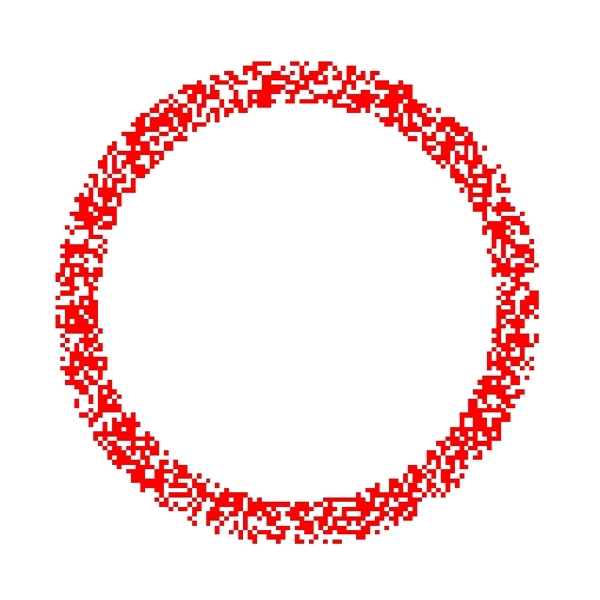}&
\includegraphics[width=0.05\textwidth ]{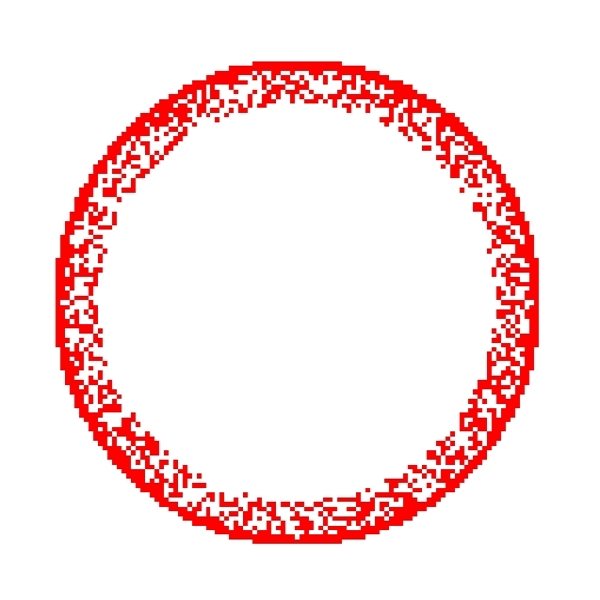}&
\includegraphics[width=0.4\textwidth ]{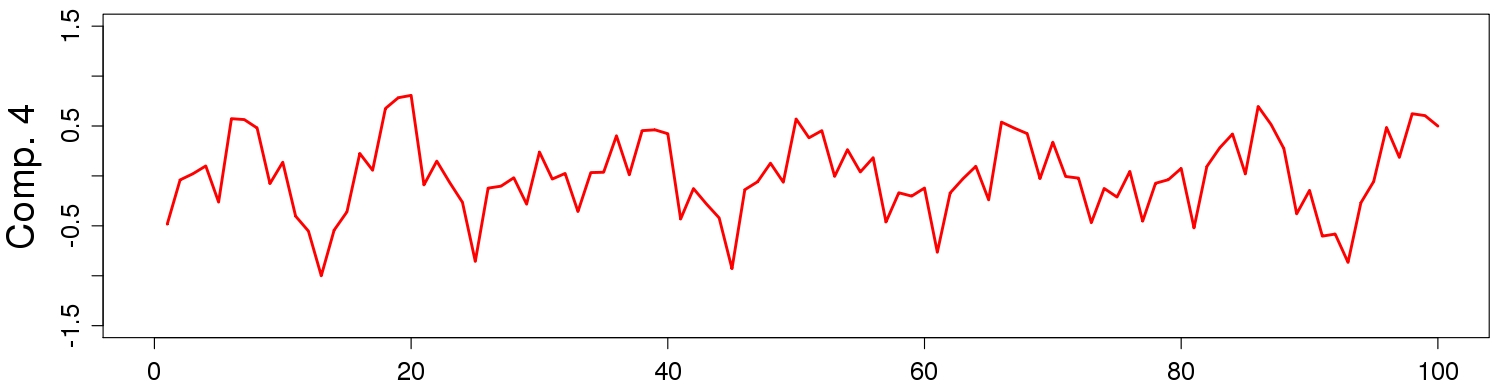}\\%
\includegraphics[width=0.4\textwidth ]{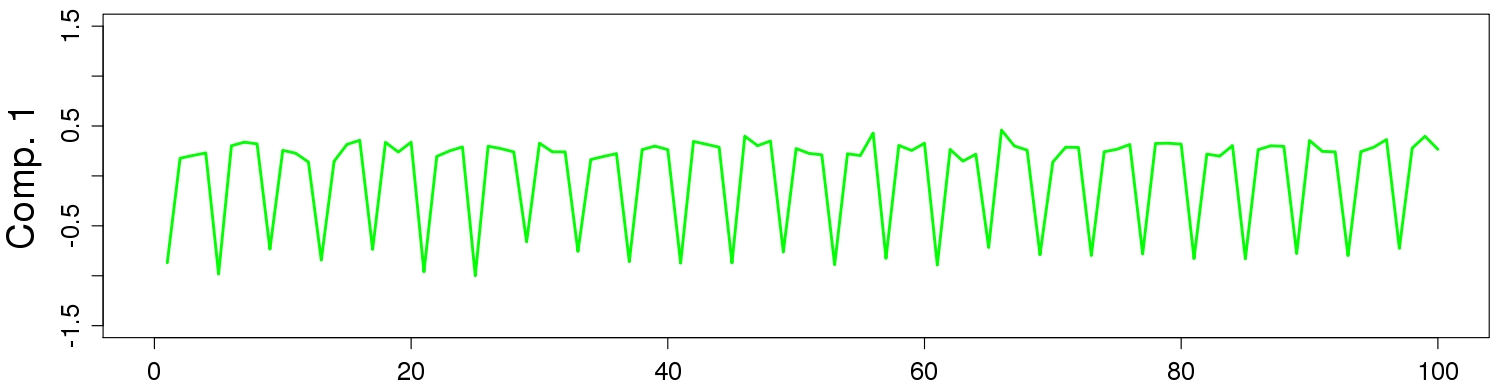}&
\includegraphics[width=0.05\textwidth ]{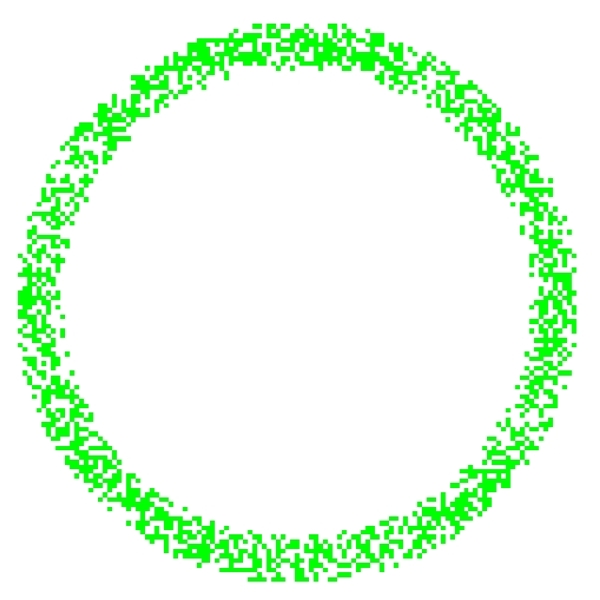}&
\includegraphics[width=0.05\textwidth ]{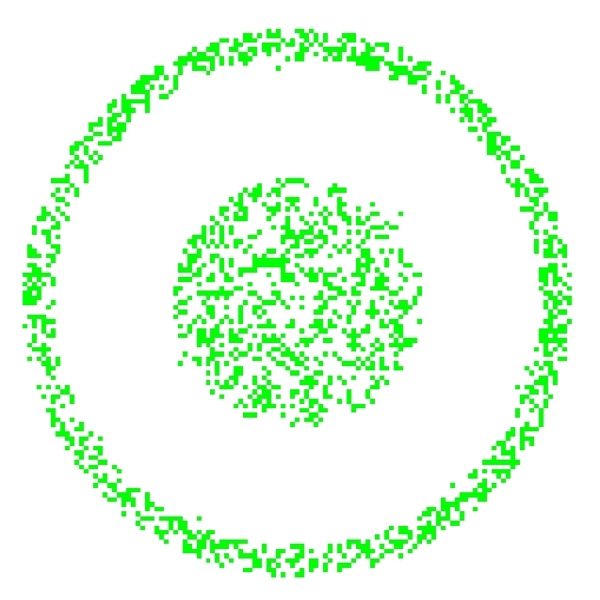}&
\includegraphics[width=0.4\textwidth ]{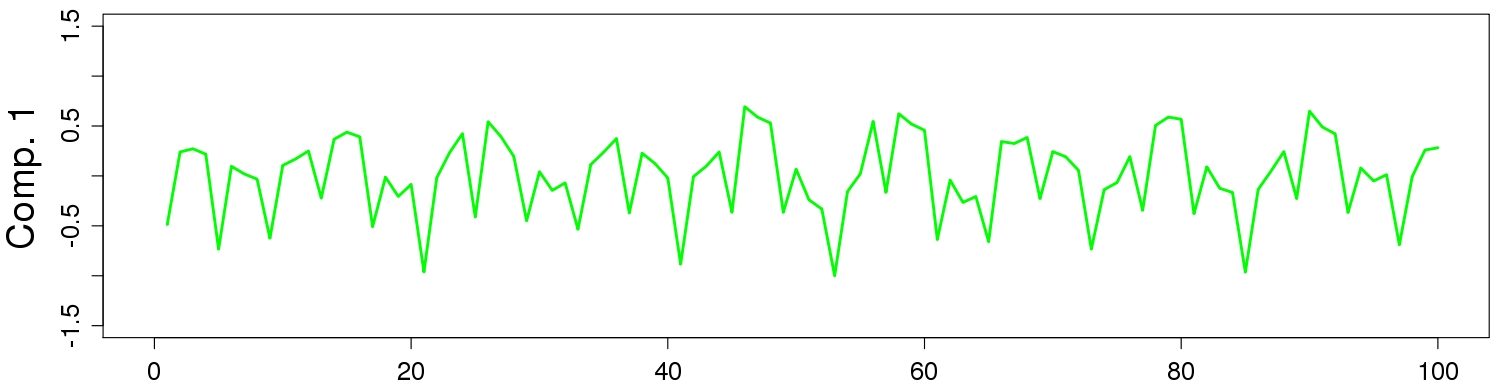}\\%
\end{tabular}
\caption{
Time course and (thresholded) localization of the extracted components obtained using temporal ICA (left) and spatial ICA (right). Each ring indicates the thresholded localization of the component coded with the same color. For temporal ICA, extracted components are similar to the simulated signals with frequency (from top to bottom) respectively of 1/11~Hz, 1/10~Hz, 1/16~Hz and 1/4~Hz and correctly localized. For spatial ICA, extracted components are noisy and found only in non overlapping regions. The frequencies of the time courses of components 1 to 4, respectively, were found to be nearly equal to 1/4~Hz, 1/10~Hz, 1/16~Hz and 1/16~Hz. See Figure \ref{circle_simul1} for the correspondence with the exact position of the simulated signals. Note that each time course was normalized.}
\label{circle_simul1_result}
\end{figure}

\subsubsection{Event-related simulation}
With event-related paradigm, neuroscientists search for voxels activated specifically by each type of stimulus. To perform a simulation in this context, we used 100 3D-images ($128\times128\times3$ voxels) composed of four \textbf{non-overlapping} and concentric tubes. Each tube contained a temporal sequence of Bernoulli random variables with various probabilities of success (see Figure \ref{circle_simul2}). The background, which surrounds the tube at the periphery, contained a Gaussian noise (sd=0.2). To be realistic, we also added everywhere a Gaussian noise (sd=0.1).\\
\begin{figure}[h!]
\centering
\begin{tabular}{cc}
\includegraphics[width=0.2\textwidth ]{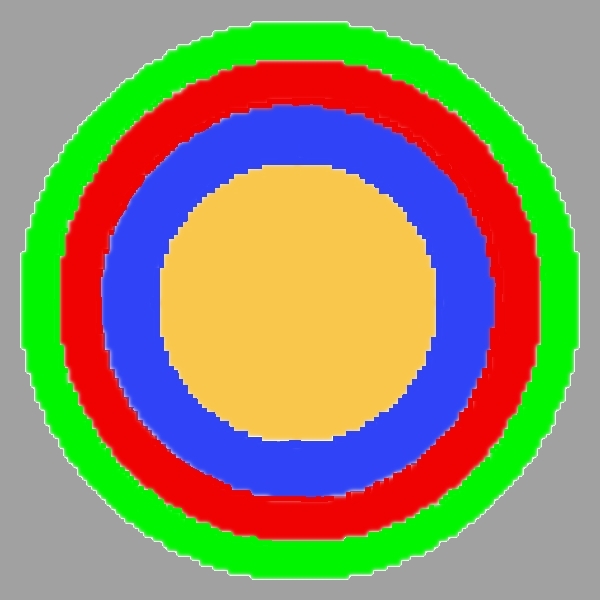}&
\includegraphics[width=0.65\textwidth ]{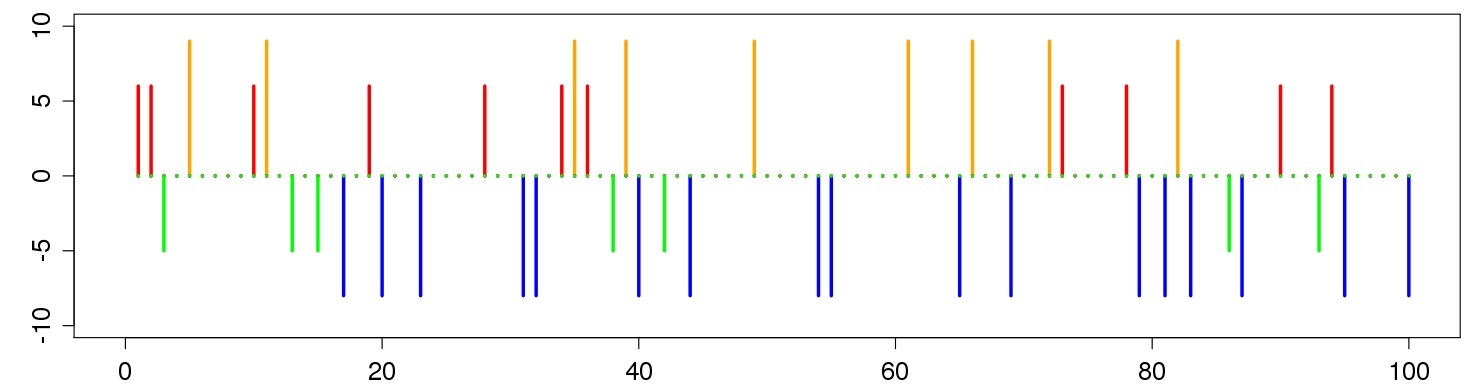}
\end{tabular}
\caption{Simulated data set. Left: A transverse slice of the volume. Each color indicates the localization of each signal. Right: Time course of each temporal sequence of Bernoulli trials displayed in their corresponding color: source 1, orange, 9 events; source 2, blue, 17 events; source 3, red, 11 events; and source 4, green, 7 events.}
\label{circle_simul2}
\end{figure}

We applied temporal and spatial ICA to these simulated data. Figure \ref{circle_simul2_result} shows the time course and  (thresholded) spatial localization of the 4 extracted components. For the latter, we used the following procedure. For each extracted time course $C_i$ ($1\leq i\leq 4$), we considered either its positive part or its negative part, selecting the one having the highest peak of amplitude (in absolute value). Let's note $\tilde{C}_i$ the selection. Then, we computed the \textit{binary correlation} (see equ. (\ref{bcor}) below) between each one of the original temporal signals of sources $S_j$ ($1\leq j\leq 4$) and a thresholded version $\tilde{C}_{i[j]}$ of $\tilde{C}_i$. The thresholds used to obtain $\tilde{C}_{i[j]}$, $1\leq j\leq 4$ were respectively 0.91, 0.83, 0.89 and 0.93 for the sources from the center to the periphery (see Figure \ref{circle_simul2}).  Intuitively, these thresholds correspond to the number of peaks of each original temporal signal among 100, i.e.  9 for source 1 (orange), 17 for source 2 (blue), 11 for source 3 (red) and 7 for source 4 (green). We define the binary correlation (number in $[-1,1]$) between two (non necessarily positive) binary random sequences $u=(u_t,1\leq t\leq T)$ and $v=(v_t,1\leq t\leq T)$ by:
\begin{equation}\label{bcor}
\textrm{bcor}(u,v)=\frac{\sum_{t=1}^T\textrm{sign}(u_t\times v_t)}{\sum_{t=1}^T\left(\textrm{sign}|u_t|+\textrm{sign}|v_t|-\textrm{sign}|u_t\times v_t|\right)}.
\end{equation}
Note that $\textrm{sign}(0)=0$.
\begin {figure}[h!]
\centering
\hspace*{-0.5cm}
\begin{tabular}{cccc}
\multirow{2}{*}{\includegraphics[width=0.35\textwidth ]{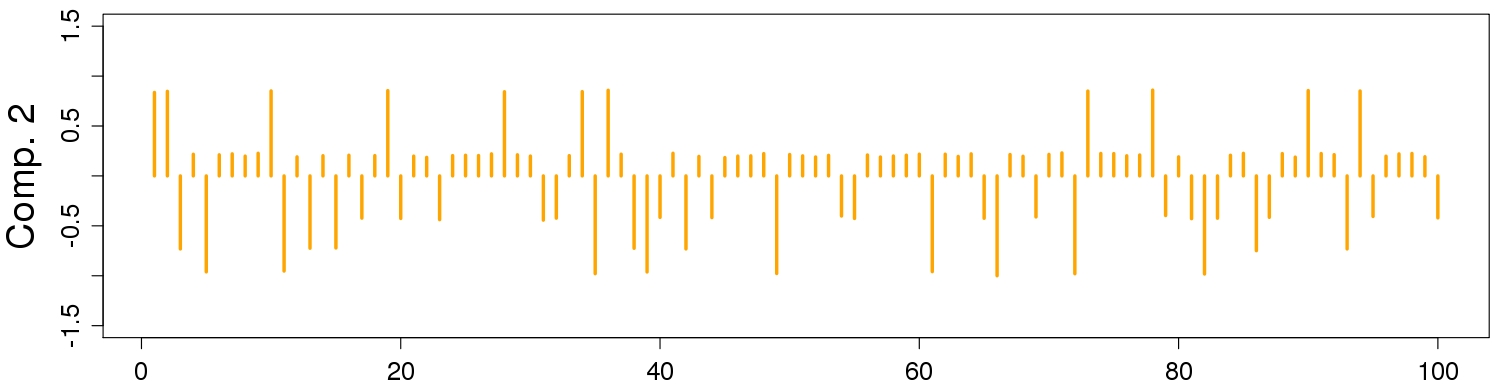}}& &
&\multirow{2}{*}{\includegraphics[width=0.35\textwidth ]{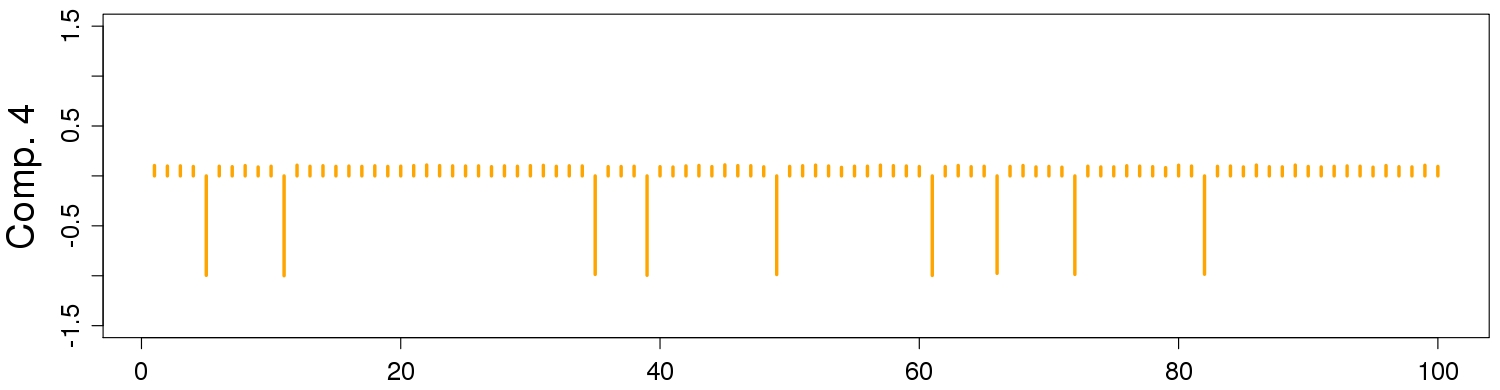}}\\%
&\includegraphics[width=0.05\textwidth ]{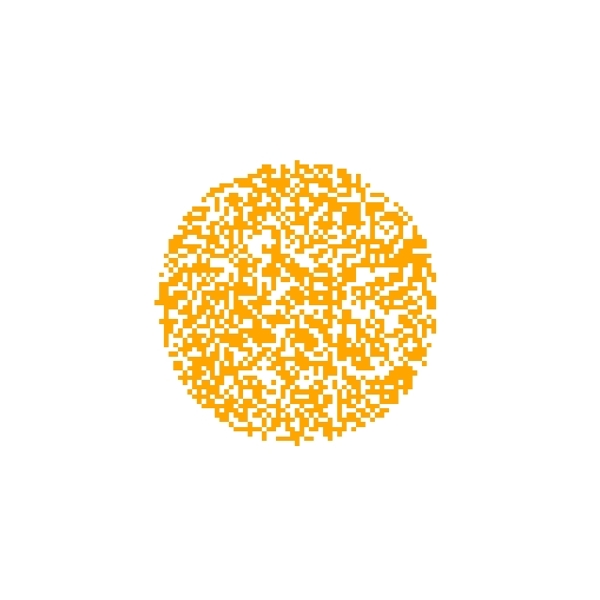}
&\includegraphics[width=0.05\textwidth ]{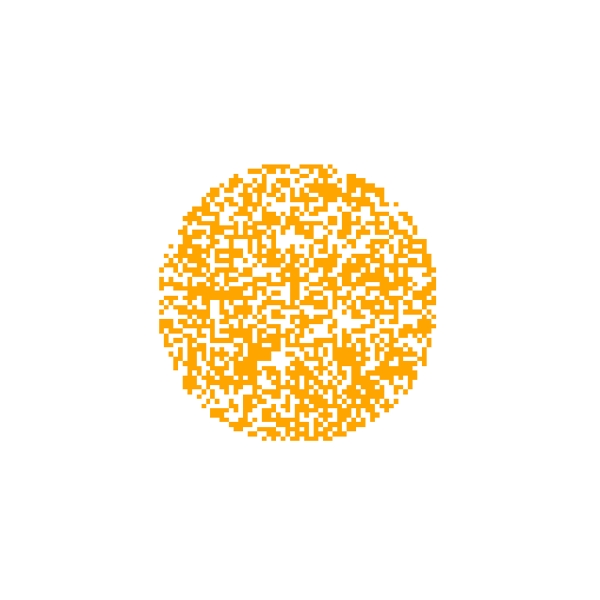}&\\
&\begin{footnotesize}$\text{bcor}_{-}=-1$ \end{footnotesize}&\begin{footnotesize} $\text{bcor}_{-}=-1$ \end{footnotesize}&\\
\multirow{2}{*}{\includegraphics[width=0.35\textwidth ]{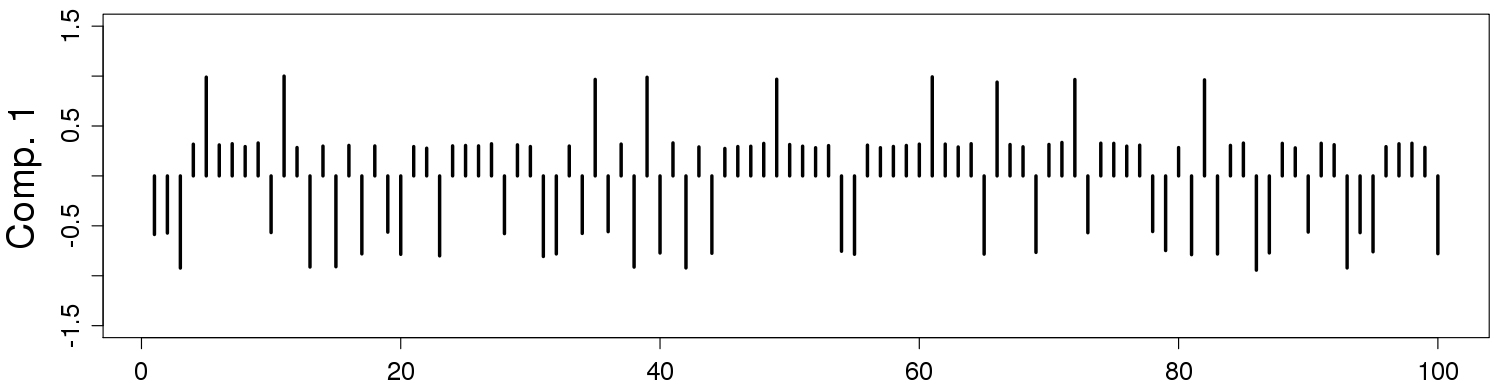}}& &
&\multirow{2}{*}{\includegraphics[width=0.35\textwidth ]{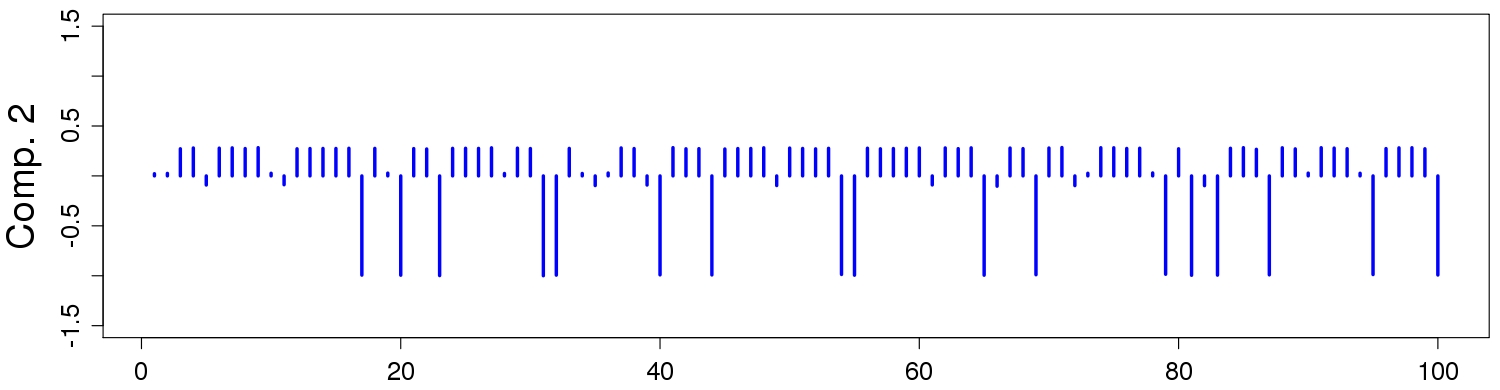}}\\%
&%\includegraphics[width=0.05\textwidth ]{SimulEvent/ICAt_Composante2_image.jpg}
&\includegraphics[width=0.05\textwidth ]{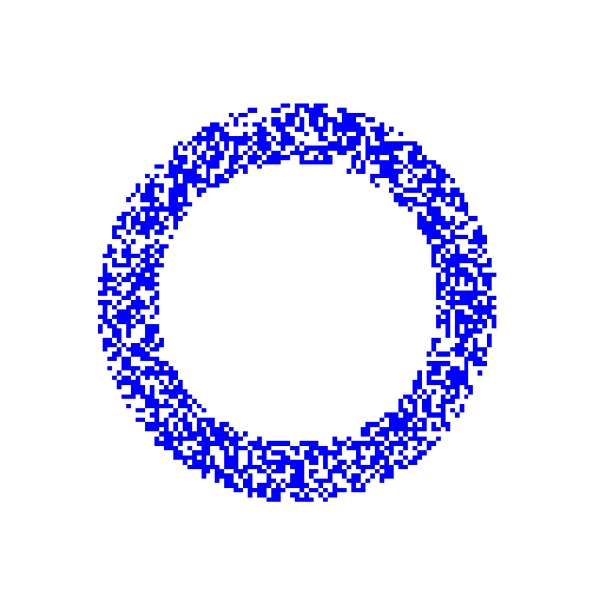}&\\
& &\begin{footnotesize} $\text{bcor}_{-}=+1$ \end{footnotesize}&\\
\multirow{2}{*}{\includegraphics[width=0.35\textwidth ]{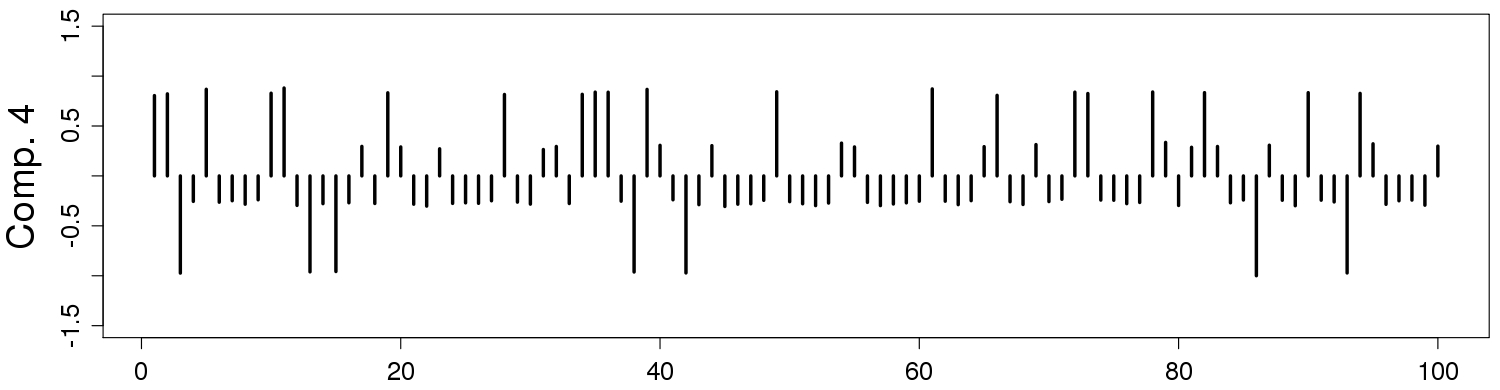}}&&
&\multirow{2}{*}{\includegraphics[width=0.35\textwidth ]{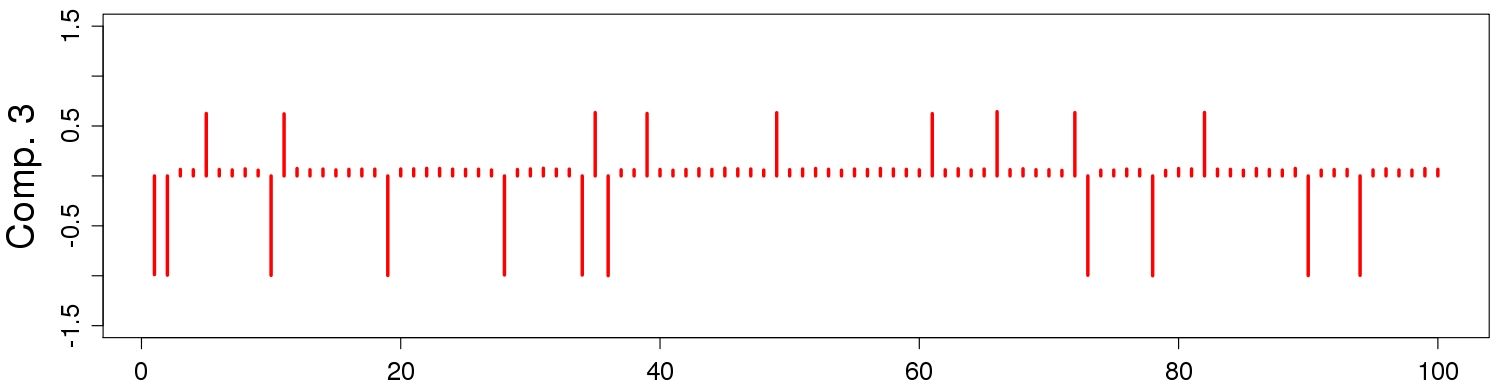}}\\%
&%\includegraphics[width=0.05\textwidth ]{SimulEvent/ICAt_Composante3_image.jpg}
&\includegraphics[width=0.05\textwidth ]{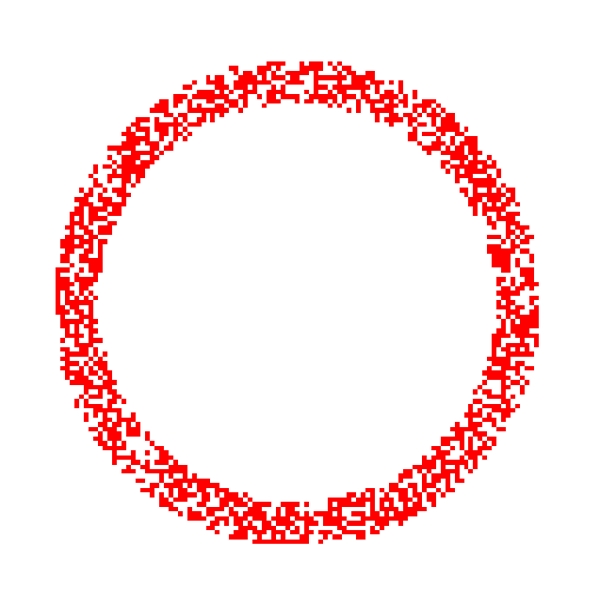}&\\
& &\begin{footnotesize} $\text{bcor}_{-}=-1$ \end{footnotesize} &\\
\multirow{2}{*}{\includegraphics[width=0.35\textwidth ]{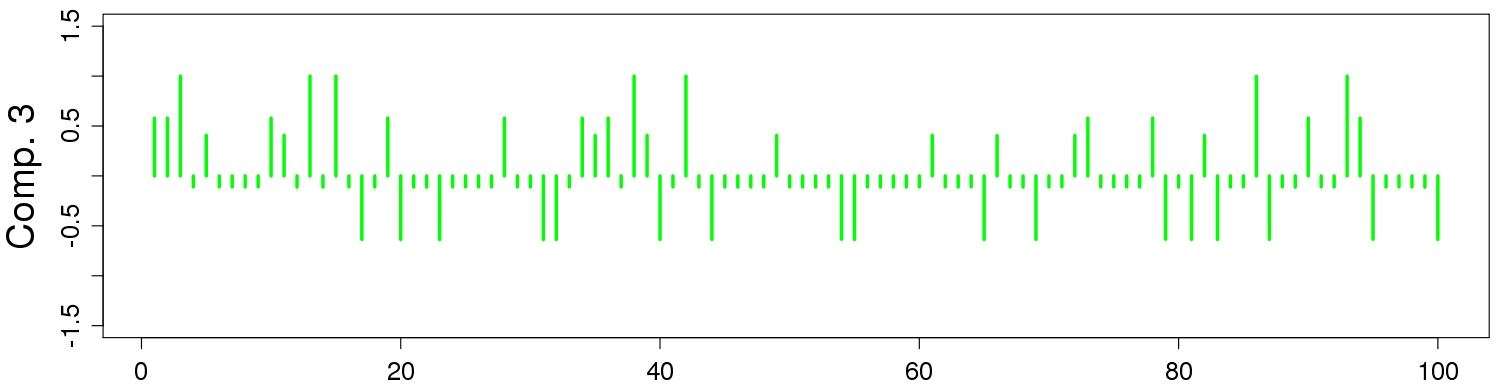}}&&
&\multirow{2}{*}{\includegraphics[width=0.35\textwidth ]{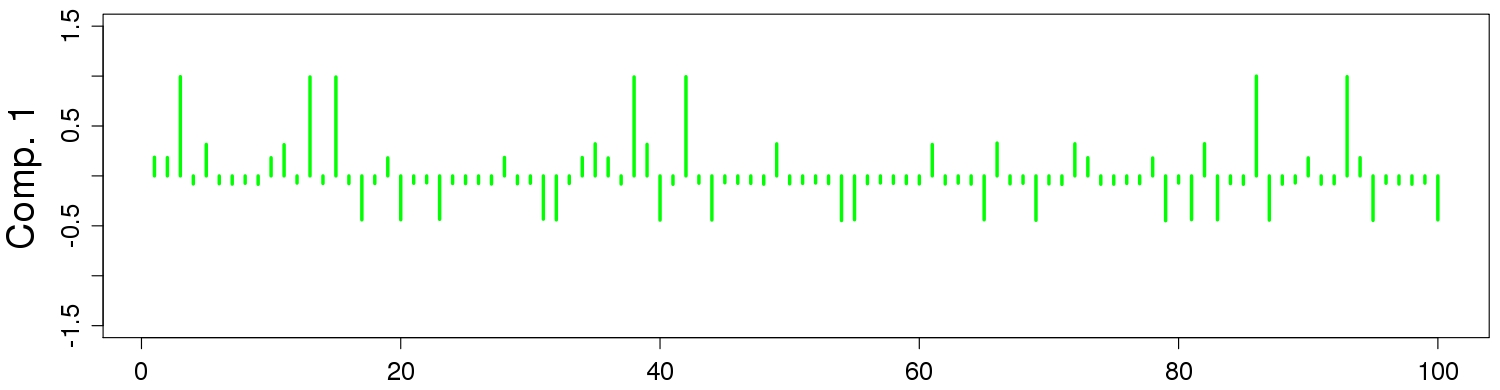}}\\%
&\includegraphics[width=0.05\textwidth ]{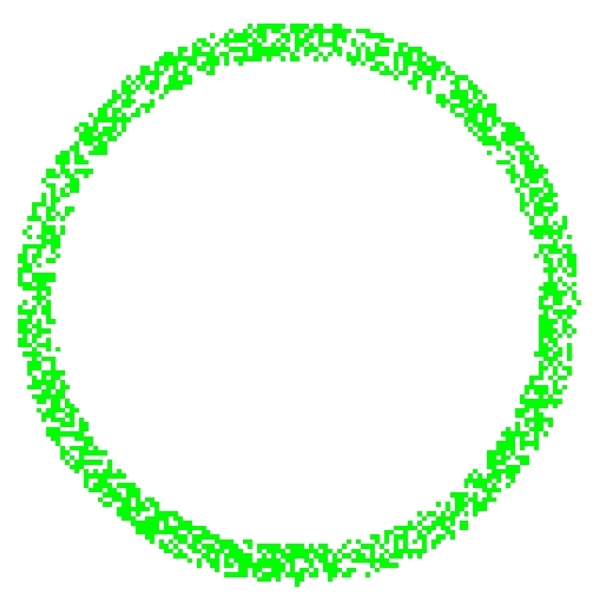}
&\includegraphics[width=0.05\textwidth ]{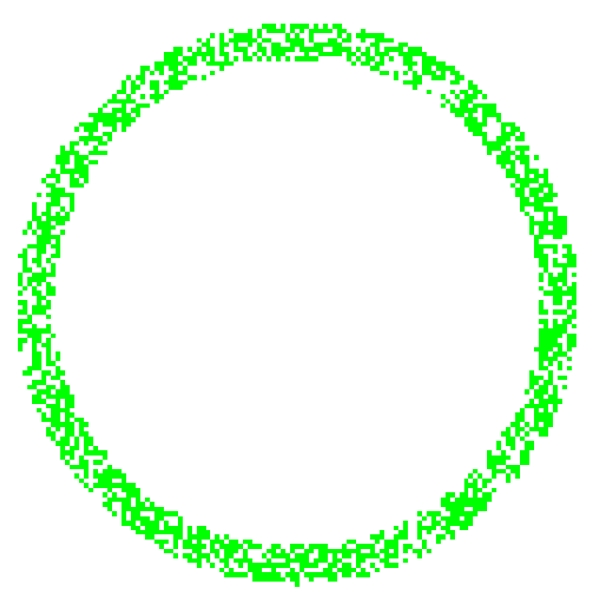}&\\
&\begin{footnotesize} $\text{bcor}_{+}=-1$ \end{footnotesize}&\begin{footnotesize} $\text{bcor}_{+}=-1$ \end{footnotesize}&\\
\end{tabular}
\caption{Time course and (thresholded) localization of the components detected using temporal ICA (left) and spatial ICA (right). Each ring indicates the localization of the component coded with the same color.  For each component, the binary  correlation coefficient of its time course with the corresponding initial signal is indicated. See Figure \ref{circle_simul2} for the correspondence with the exact position of the simulated signals. See text for the computation of the components localization.}
\label{circle_simul2_result}
\end{figure}

We then assigned each extracted component to the original signal corresponding to the computation of the highest absolute value of the binary correlation (see Figure \ref{circle_simul2_result}). For tICA, we found the following results:
\begin{itemize}
 \item Components 1 and 2 were assigned with the temporal signal of source 1, with a binary correlation equal respectively to +1 and -1. As there was a conflict between the spatial localization given by these two components, we computed an ``energy" index as follows. Let $\tilde{C}_1$ and $\tilde{C}_2$ be the two parts selected from the components 1 and 2 respectively. We then divide $\tilde{C}_1$ and $\tilde{C}_2$ respectively by $\underset{1\leq t\leq T}{\max}\tilde{C}_1$ and $\underset{1\leq t\leq T}{\max}\tilde{C}_2$ to obtain $C_1^*$ and $C_2^*$.  Then we threshold $C_1^*$ and $C_2^*$  using the threshold $t_{12[1]}$ which is equal to half the empirical quantile of order 0.91 of the temporal signal $|C_1^*|+|C_2^*|$. The ``energy'' of component 1 versus component 2 to explain the source 1 is then given by the sum of the values in $|C_1^*|$ above $t_{12[1]}$. Similarly, the ``energy'' associated with component 2 is given by the sum of the values in $|C_2^*|$ above $t_{12[1]}$. The ``energy" index was higher for component 2 (ratio of 0.58) which was consequently assigned to the temporal signal of source 1.
 
 \item Components 3 and 4 were assigned with the temporal signal of source 4, with a binary correlation equal respectively to -1 and +1. Here again, we computed the ``energy" index which was higher for component 3 (ratio of 2.5) thus assigned to the temporal signal of source 4. 
\item Components 1 and 4 were consequently not associated with any source. 
\end{itemize}

For sICA, we found the following results:
\begin{itemize}
 \item Component 1 was assigned with the temporal signal of source 4, with a binary correlation equal to -1. 
 \item Component 2 was assigned with the temporal signal of source 2, with a binary correlation equal to +1.
 \item Component 3 was assigned with the temporal signal of source 3, with a binary correlation equal to -1. 
 \item Component 4 was assigned with the temporal signal of source 1, with a binary correlation equal to -1.
\end{itemize}

Surprisingly, spatial ICA works better in this case as compared to temporal ICA. We checked, using the \R package \pkg{IndependenceTests}, the independence of our original random sequences of Bernoulli trials (note that this package can also check the independence of variables that are singular with respect to the Lebesgue measure). There were no reason to significantly reject this independence hypothesis (at 5\% level). On the other side, the temporal extracted components were significantly dependent. A possible explanation to the tICA failure (notwithstanding the fact that standard ICA model is only defined for continuous random variables, since the unmixing and mixing matrix coefficients are real numbers and thus are not constrained in anyway to give binary values) may be the use of kurtosis in the FastICA algorithm, a quantity which is not optimal for sequences of Bernoulli trials (see \cite{Himberg01}).

\subsubsection{Traveling wave simulation}

We generated several sinusoids with the same fundamental frequency $f$=1/16~Hz and various phases to simulate traveling activation waves. The resulting data set consisted in a sequence comprising 240 3D-images. Each image ($128\times128\times3$ voxels) was composed of four \textbf{partially overlapping} and concentric tubes.  
Each tube contained a pure sinusoidal signal (with a frequency $f$ equal to 1/16~Hz) in its non overlapping part and a sum of two pure sinusoidal signals in its parts that intersect with another tube.  For pure signals, different phases were considered, namely $\phi_1=0$, $\phi_2=\pi/4$,  $\phi_3=\pi/2$ and $\phi_4=3\pi/4$ from the tube at the center to the one at the periphery respectively. The background, which overlaps the tube at the periphery, contained, in its non overlapping part, a Gaussian noise (sd=0.2), see Figure \ref{circle_simul3}.  Thus, four pure signals and four mixed signals were present. To be realistic, we also added everywhere a Gaussian noise (sd=0.1).\\

\begin{figure}[h!]
\centering
\begin{tabular}{cc}
\includegraphics[width=0.2\textwidth ]{SimulPierre/circle.jpg}&
\includegraphics[width=0.65\textwidth ]{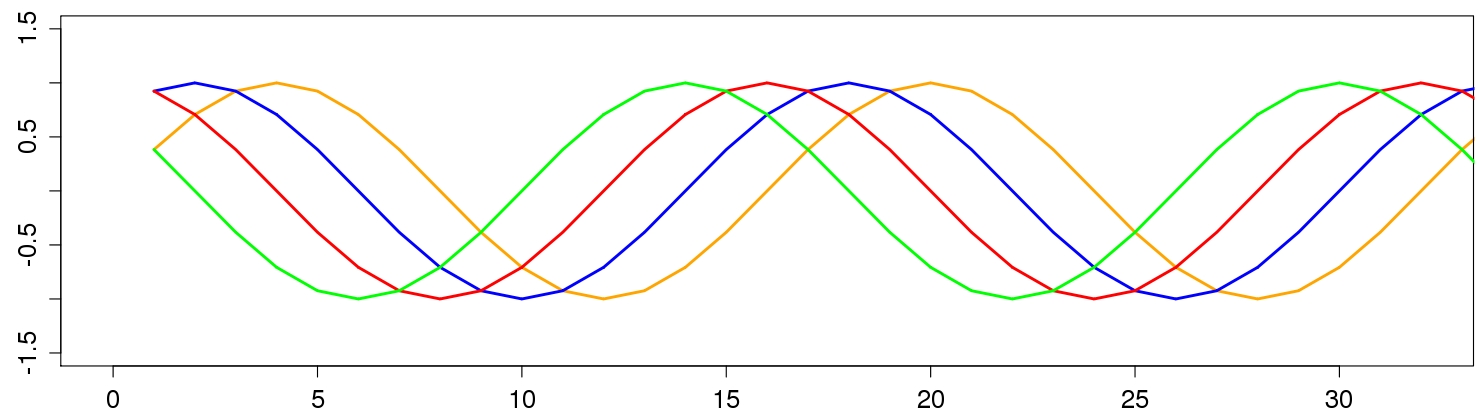}\\
\end{tabular}
\caption{Simulated data set. Left: A transverse slice of the volume. Each color indicates the localization of each signal. Pure signals are represented with a color (orange, blue, red and green),  Gaussian noise (sd=0.2) is in grey and mixed signals are in white. Right: Temporal course of the pure single signals written in their corresponding color ($f$=1/16~Hz, phases = 0, $\pi/4$, $\pi/2$, $3\pi/4$ respectively from the center to the periphery).}
\label{circle_simul3}
\end{figure}

Before going any further, it is convenient to think about a sinusoid waveform, which is deterministic in nature, as a sequence of different realizations of the same random variable $X=\sin(2\pi Uf+\phi)$, where $U$ is a continuous uniform random variable or, even better in the present case, a discrete uniform random variable on the sampled points. Note that, with standard algorithms, blind source separation is not concerned with the sequencing of the input signals. Indeed, changing the time ordering in which the mixtures are presented at the input will always lead to the same source separation (with the corresponding change in time indexing). This comment also applies to the two previous simulations. Note also that sinusoids with the same frequency but presenting different phases are in fact not independent. For example, correlation is not zero except for sinusoids with phase difference of $\pi/2$. Indeed, let $X=\sin(2\pi Uf+\phi_1)$ and $Y=\sin(2\pi Uf+\phi_2)$ be two random variables, where $U$ is a discrete uniform random variable with support $\{a,a+1,\ldots,b-1,b\}$, i.e. with characteristic function $\varphi_U(t)=\frac{e^{iat}}{n}\sum_{k=0}^{n-1}e^{ikt}$ where $n=b-a+1$. We have been able, after tedious computations, to obtain explicitly the covariance $\mathbb{C}\textrm{ov}(X,Y)=\mathbb{E}(XY)-\mathbb{E}(X)\mathbb{E}(Y)$ between $X$ and $Y$ by showing that $$\mathbb{E}(X)=\textrm{Im}\left[e^{i\phi_1}\frac{e^{ia2\pi f}}{n}\sum_{k=0}^{n-1}e^{ik2\pi f}\right]=\frac{1}{n}\sum_{k=0}^{n-1}\sin\left(\phi_1+2\pi f(a+k)\right)$$ 
and
$$\mathbb{E}(XY)=\frac{1}{2}\left[\cos(\phi_1-\phi_2)-\frac{1}{n}\sum_{k=0}^{n-1}\cos\left(\phi_1+\phi_2+4\pi f(a+k)\right)\right].$$

Temporal and spatial ICA  extracted 3 components. As expected, temporal ICA extracted two components corresponding to sinusoids with phase difference of $\pi/2$. Spatial ICA does not impose any (independence) constraint on the extracted time courses. This is reflected in the results obtained for sICA.

\begin {figure}[h!]
\centering
\begin{tabular}{cc}
\includegraphics[width=0.4\textwidth ]{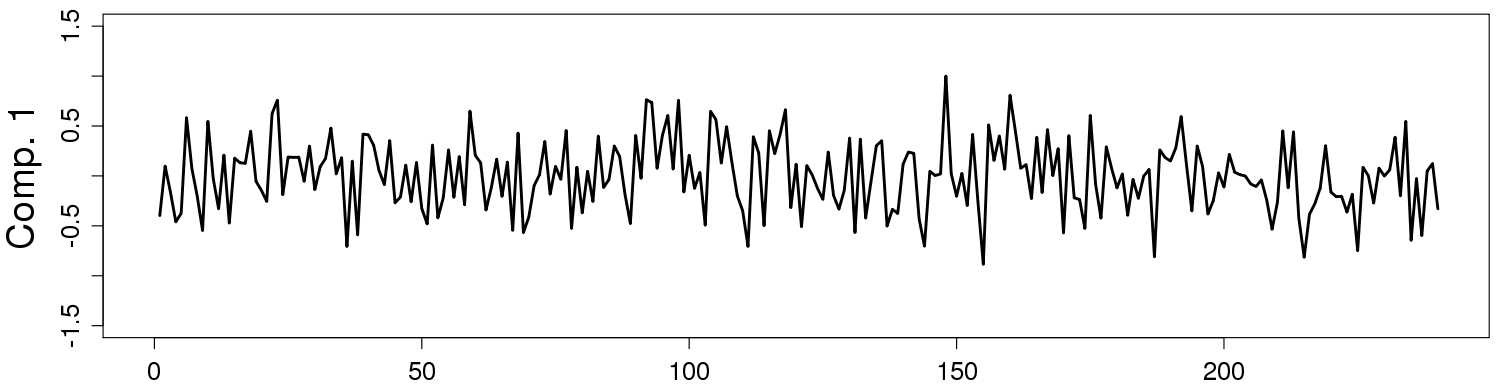}&
\includegraphics[width=0.4\textwidth ]{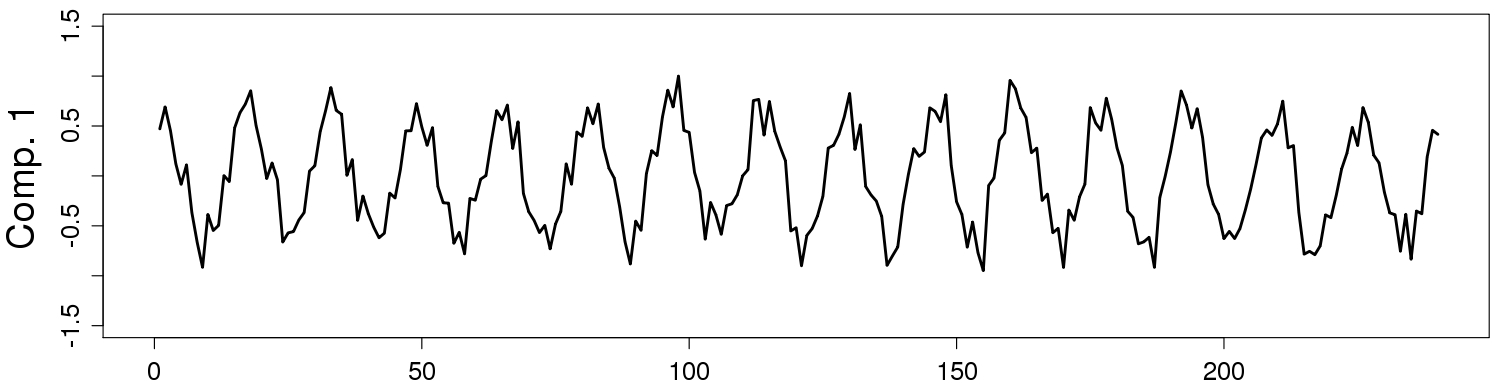}\\
 & $f = 1/16~Hz$ \\
\includegraphics[width=0.4\textwidth ]{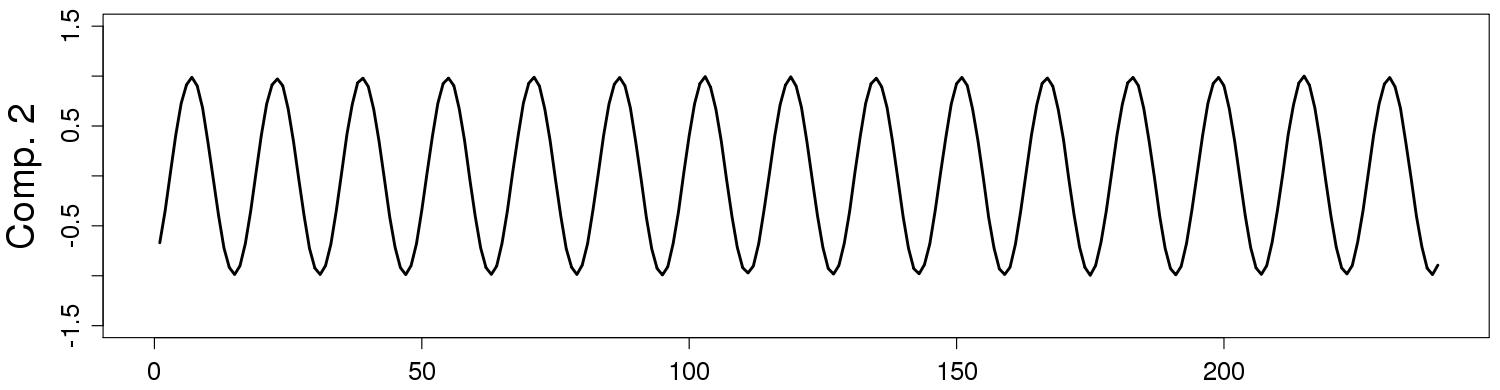}&
\includegraphics[width=0.4\textwidth ]{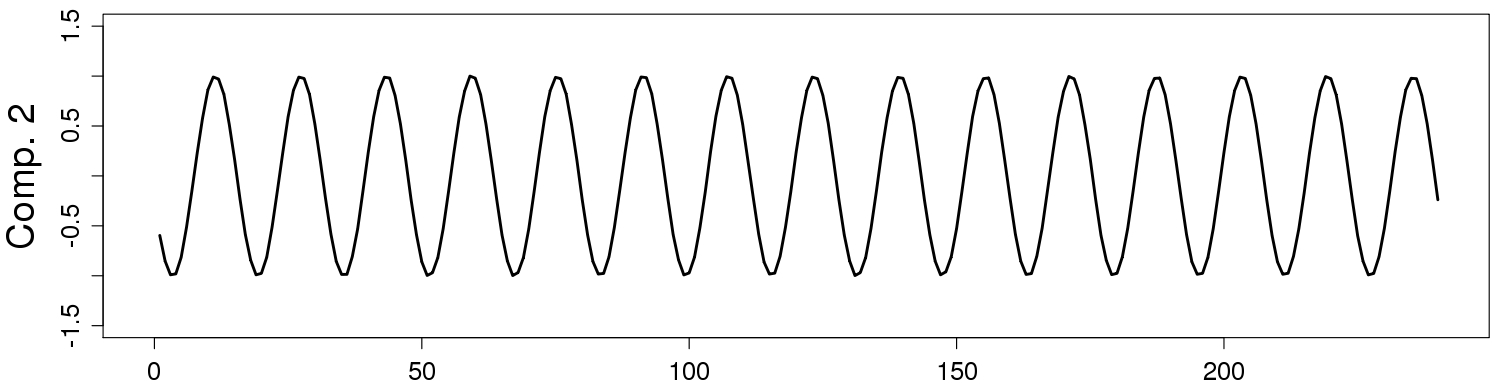}\\%
$f = 1/16~Hz$   &  $f = 1/16~Hz$ \\
\includegraphics[width=0.4\textwidth ]{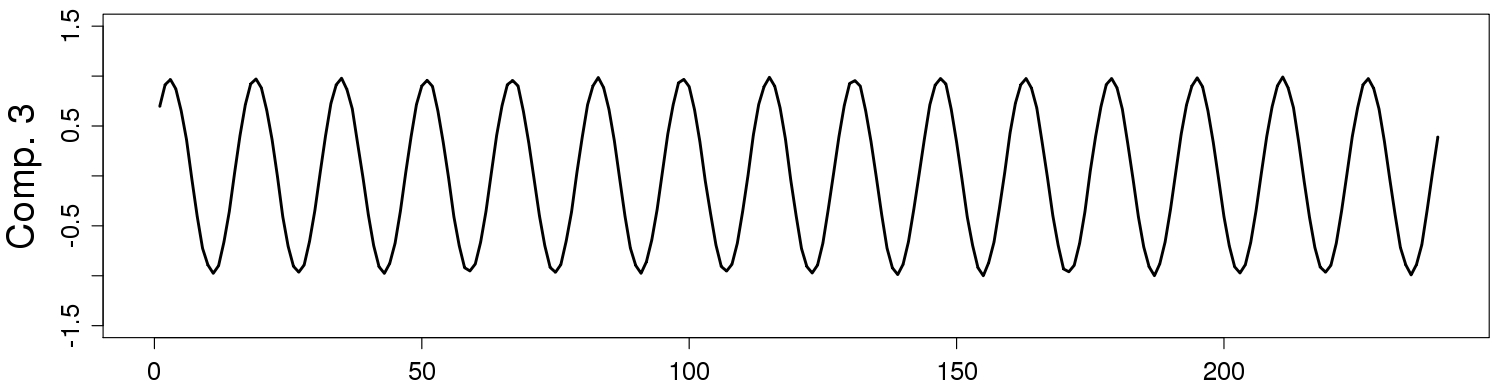}&
\includegraphics[width=0.4\textwidth ]{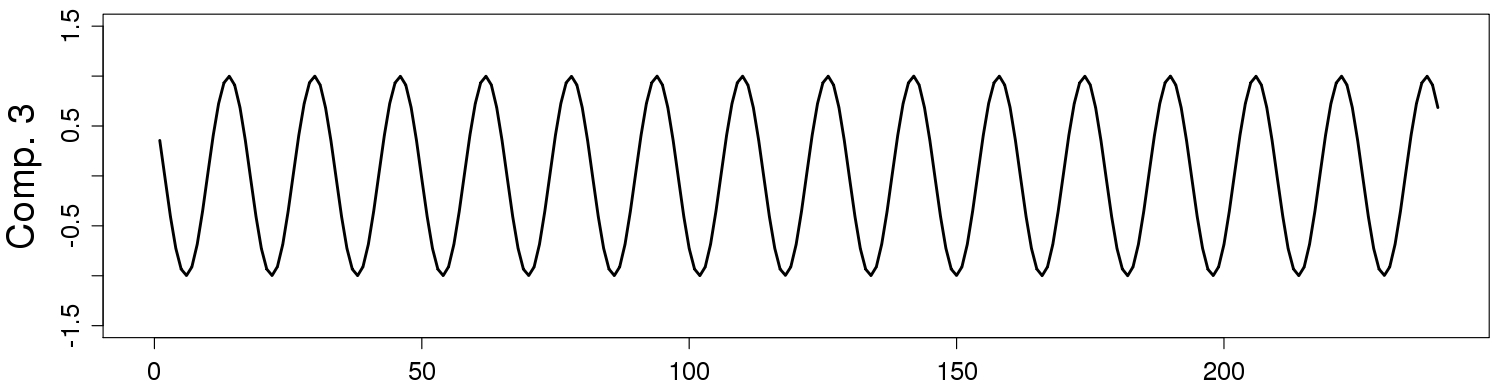}\\
$f = 1/16~Hz$ & $f = 1/16~Hz$ \\
\end{tabular}
\caption{Time course of the components detected using temporal ICA (left) and spatial ICA (right). For tICA, the phase difference between components 3 and 2 is $\pi/2$. The first component (top left) represents a noise signal. For sICA, the phase difference between the temporal signal of source 1 and, respectively, the time courses of components 1, 2 and 3 are: 1.037, 3.377 and 2.385. Note that ICA cannot recover the sign of the sources, so the phases of extracted sinusoidal signals are only defined modulo $\pi$. Note also that each component was normalized. }
\label{circle_simul3_result}
\end{figure}

\subsection{Real data sets}
We conducted two types of evaluations using real data sets coming from retinotopic mapping and color center mapping experiments. These data were part of a cognitive study investigating which color sensitive areas are specially involved with colors induced by synesthesia (\cite{Hupe10}). The real data sets used are provided as supplementary material.\footnote{The data provided should exclusively be used by the journal reviewers and readers to reproduce our examples. They cannot be used to any other purpose without the express authorization of the authors. }
\subsection*{Experiment 1 : Retinotopy mapping}
Retinotopic mapping of human visual cortex using fMRI is a well established method (\cite{Sereno95}; \cite{Warnking02}) that allows to properly delineate low visual areas. It uses four separate experiments with 4 periodic stimuli (an expanding/contracting ring and a rotating counter or anti-counter clockwise wedge) to measure respectively eccentricity and polar angle maps. For this study, we only used functional MRI data corresponding to the expanding ring experiment (240 volumes acquired each 2 seconds). The periodic visual stimulus expanded from 0.2 to 3 degrees in the visual field during 32 seconds and was repeated fifteen times. This periodic stimulation generated a wave of activation in the retinotopic visual areas (\cite{Engel94}), located in the occipital lobe, at the frequency of 1/32 Hz  measured at a discrete temporal sampling of 2 seconds (equivalent to 1/16 temporal bins).
After IC analysis of these functional data, using our \R function \code{f.icast.fmri.gui()}, 18 and 15 components were automatically extracted respectively with tICA and sICA. In this experiment, we searched for components corresponding to cortical activation at the frequency of the visual stimulation. tICA and sICA extracted more (noisy) components than the ones specific to the stimulus. Indeed, the main problem with fMRI data is that each activated voxel of each volume contains a mixture of the signal of interest (BOLD effect) with several confound signals with several origins: ocular movement, heart rate, respiratory cycle, or head movement. 
Figure \ref{retino_components} shows the temporal and spatial components at the frequency of the visual stimulation corresponding to the cortical activation of interest.  The computed phases are (approximatively) respectively equal to  $\pi/8$, $3\pi/8$ and $5\pi/8$ for tICA (green, red and blue components) and $\pi/8$ and $5\pi/8$ for sICA (green and red components). Figure \ref{retino_components}  displays on the corresponding anatomical image the cortical localization of these extracted components. For tICA and for each temporal component, this is done by selecting in the associated column of the estimated mixing matrix (see  equation (\ref{decompX})) the most active voxels, defined arbitrarily (see \cite{Beckmann04} for another approach) as those whith a value above the 95\% quantile (in absolute value). For sICA, we also thresholded arbitrarily each component at the 95\% quantile. Based on the retinotopy property of the visual system, the expanding ring generates a cortical activation wave moving from the posterior part to the anterior part of the occipital lobe. As indicated in Figure \ref{retino_components} the computed phases of the extracted components increase as expected from the posterior to the anterior part of the occipital lobe.  

\begin{figure}
\centering
\begin{tabular}{cccc}
 %Temporal ICA components & Spatial ICA components\\

\multirow{3}{*}{\includegraphics[width=0.2 \textwidth ]{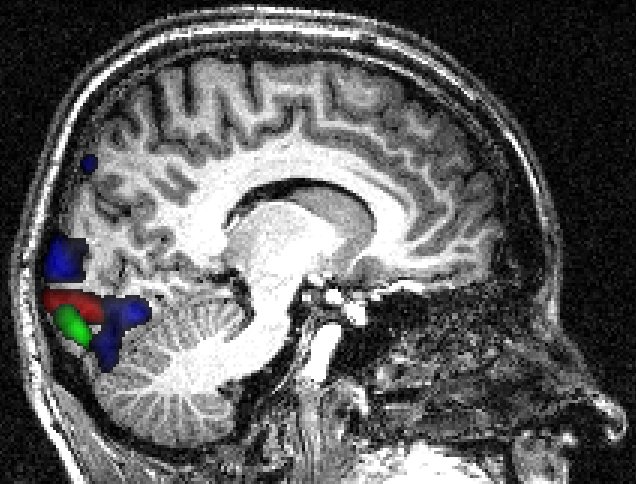}}
&\includegraphics[width=0.25\textwidth ]{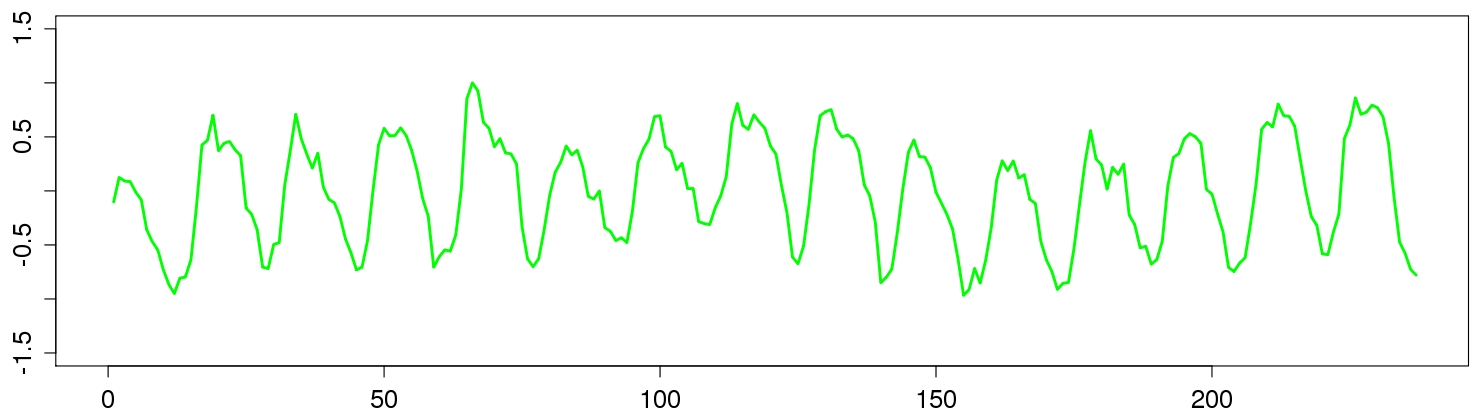}&
\multirow{3}{*}{\includegraphics[width=0.2 \textwidth ]{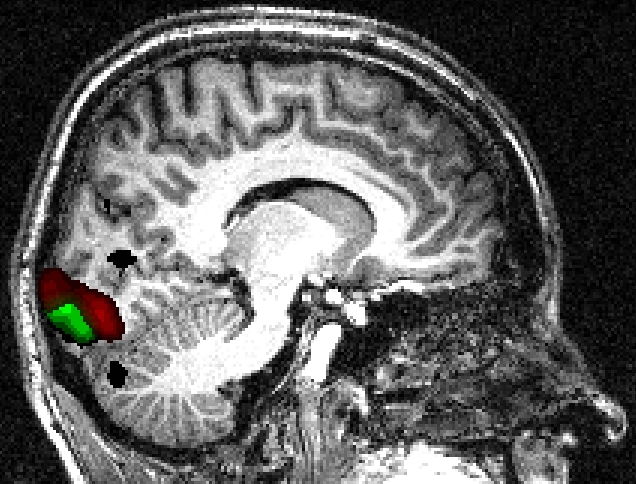}}&
 %3 extracted components at 16Hz & 2 extracted components at 16Hz\\
\includegraphics[width=0.25 \textwidth ]{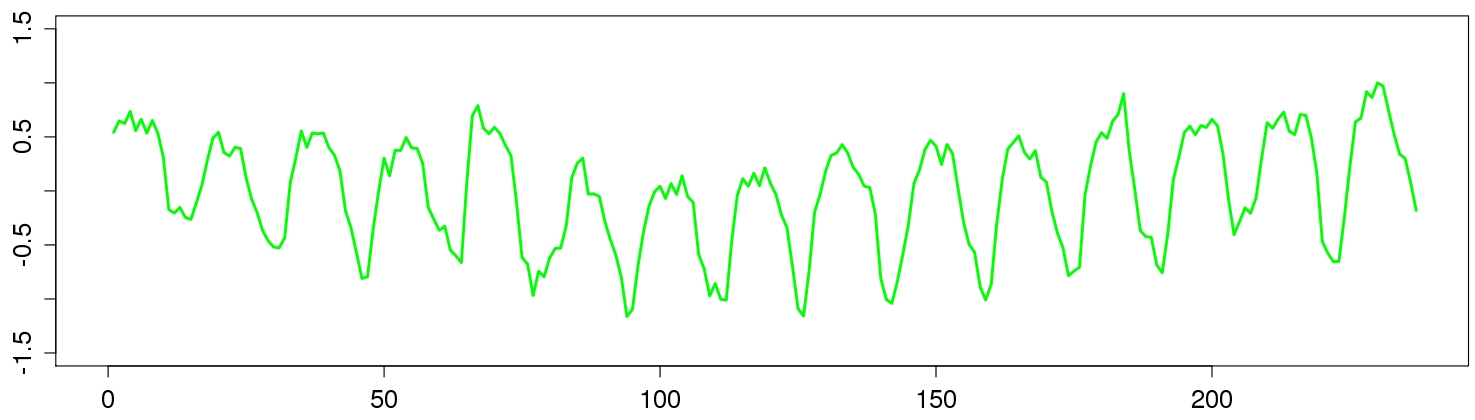}\\
&\includegraphics[width=0.25 \textwidth ]{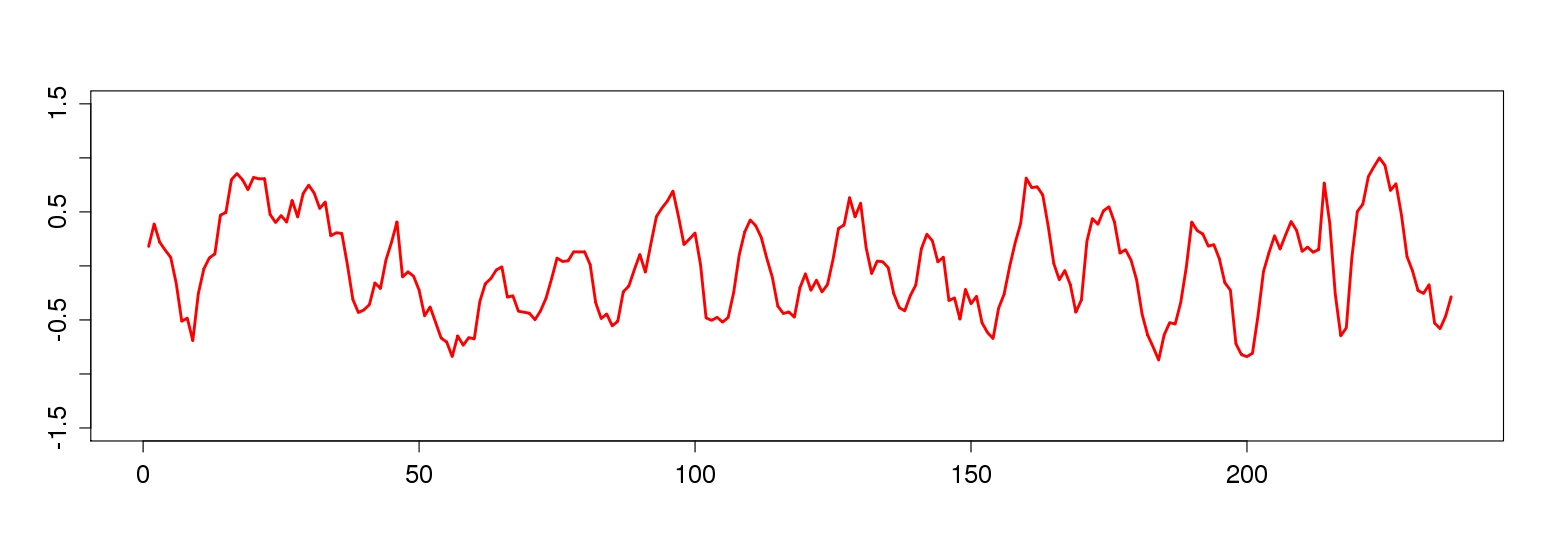}&
&\includegraphics[width=0.25 \textwidth ]{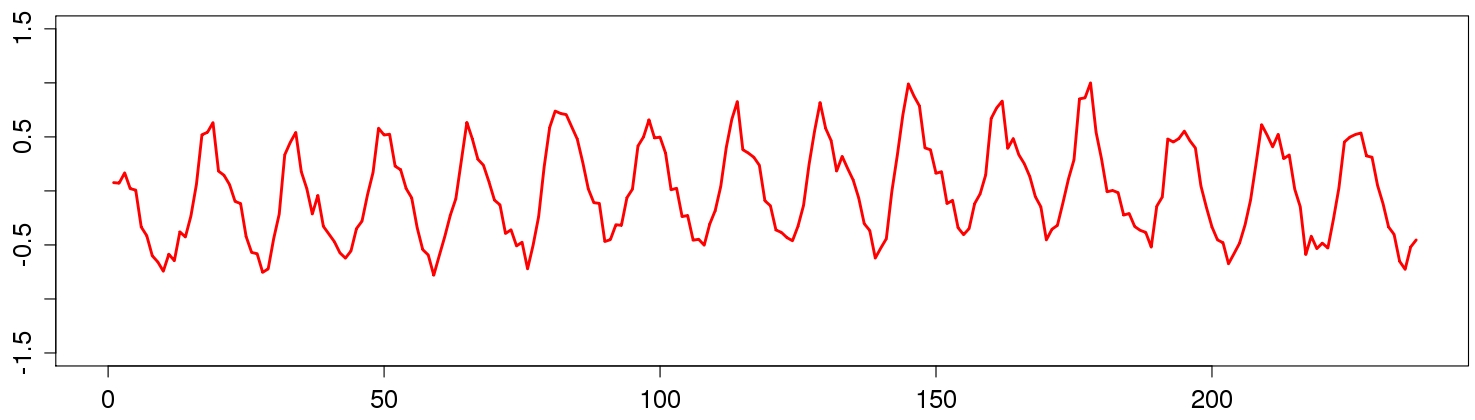}\\
&\includegraphics[width=0.25 \textwidth ]{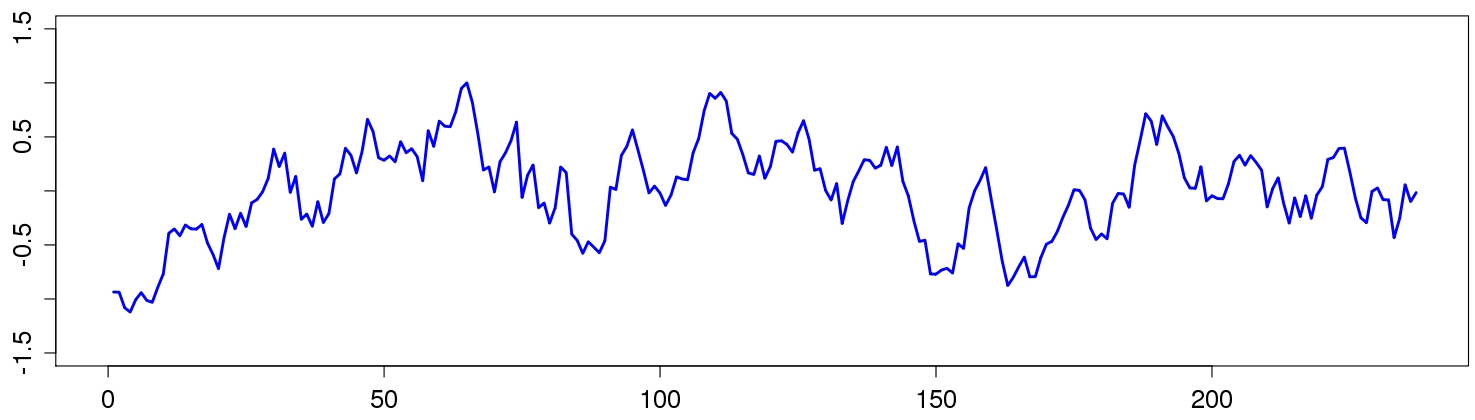}&&\\
\end{tabular}
\caption{Time course of extracted components and their spatial localization. Left: extracted components using tICA. The phases of these components are found to be (approximatively) $\pi$, $3\pi/8$, and $5\pi/8$ for green, red and blue components respectively. Right: extracted components using sICA. The phases of these components are found to be (approximatively) $\pi/8$ and $5\pi/8$ for green and red components respectively. On the anatomical MR scan (sagittal view) is indicated the localization of the most activated voxels for each component displayed with the corresponding color. The sequencing of the activated voxels follows the retinotopic property of the visual system: the periodic visual stimulation, an expanding ring, generates a periodic cortical activation moving from the posterior to the anterior part of the occipital lobe.}
\label{retino_components}
\end{figure}

\subsection*{Experiment 2 : Color center mapping}
In this experiment, we presented to the subject two stimuli, a set of chromatic rectangles (Mondrian like patterns) and the same patterns in an achromatic version (Figure \ref{mond_results}).  Each chromatic and achromatic sets of rectangles were periodically presented during successive blocks of 10 seconds. Our analysis was made on 120 functional volumes acquired each 2 seconds. Because we were only interested in visual areas, we used a mask to select only the occipital part of each volume. Using our \R function \code{f.icast.fmri.gui()}, 21 and 20 components were automatically extracted with tICA and sICA respectively. As shown in Figure \ref{mond_results}, we found both with tICA and sICA one periodic component at the frequency of the stimulus. 

\begin{figure}[h!]
\centering
 \begin{tabular}{cp{0.5cm}ccc}
%\multicolumn{3}{c}{Original Stimuli} & Temporal ICA component & Spatial ICA component\\
\includegraphics[width=0.1\textwidth ]{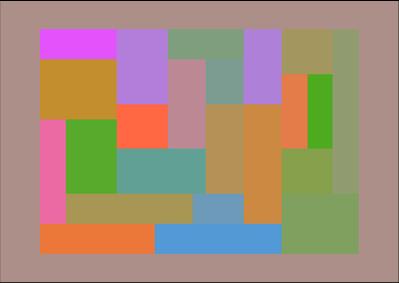}&vs&\includegraphics[width=0.1\textwidth ]{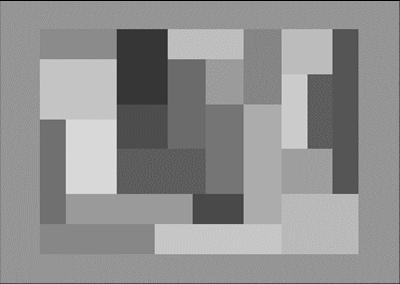} &  & \\
\multicolumn{3}{c}{\includegraphics[width=0.3 \textwidth ]{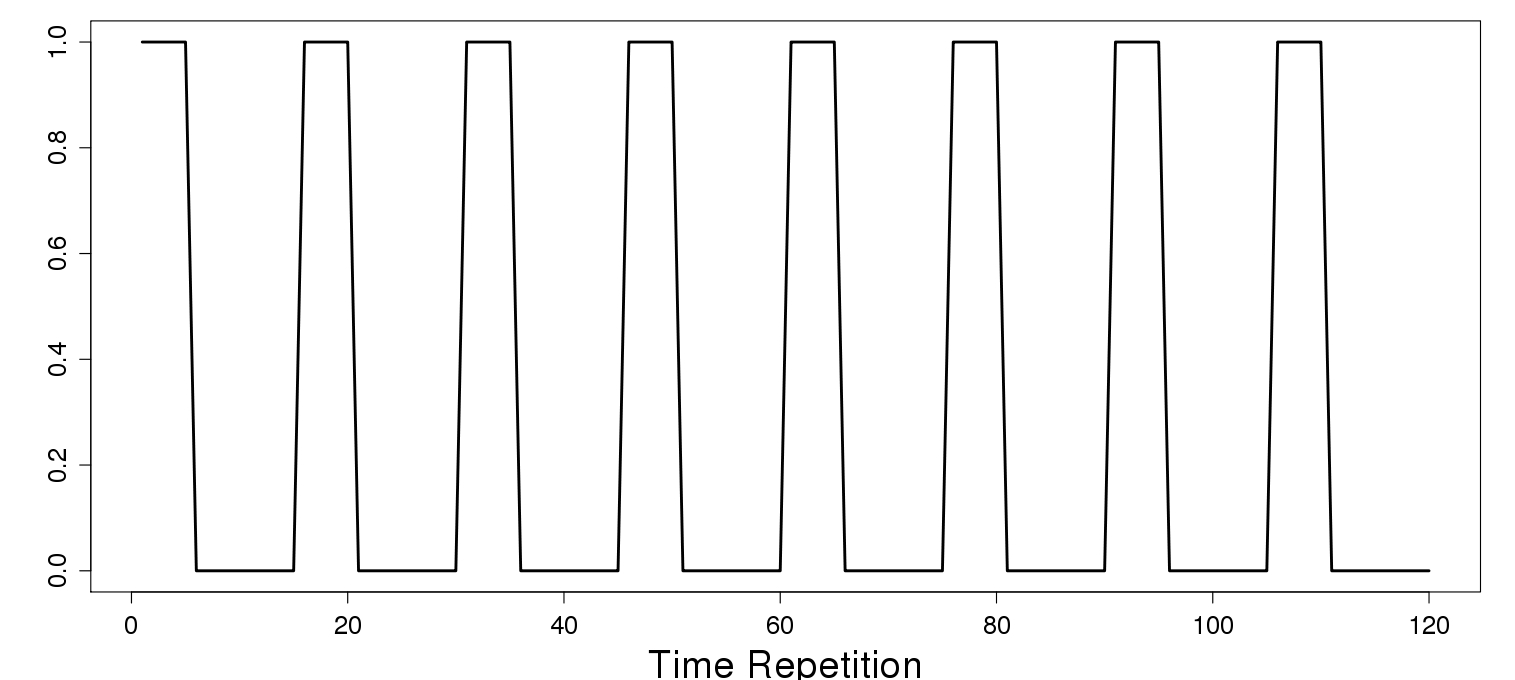}}&
\includegraphics[width=0.3 \textwidth ]{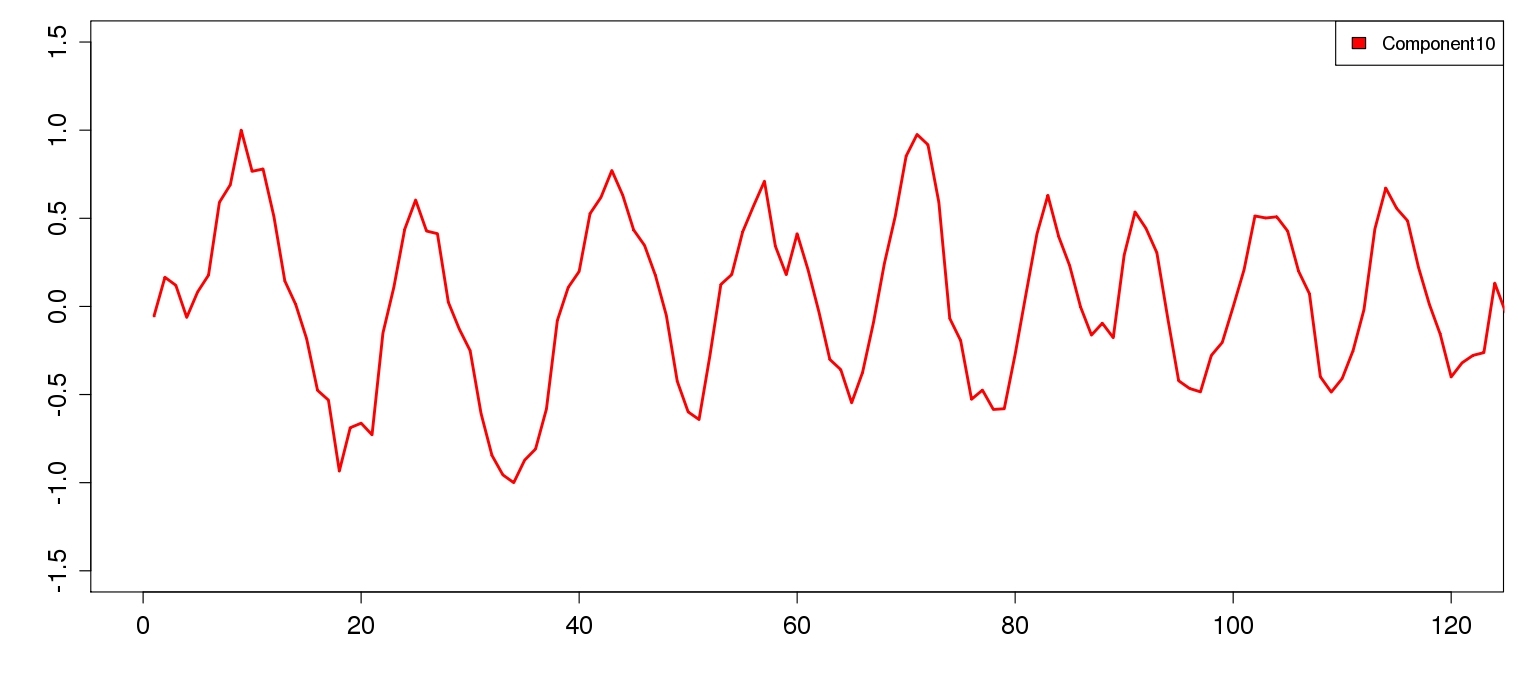}&
\includegraphics[width=0.3 \textwidth ]{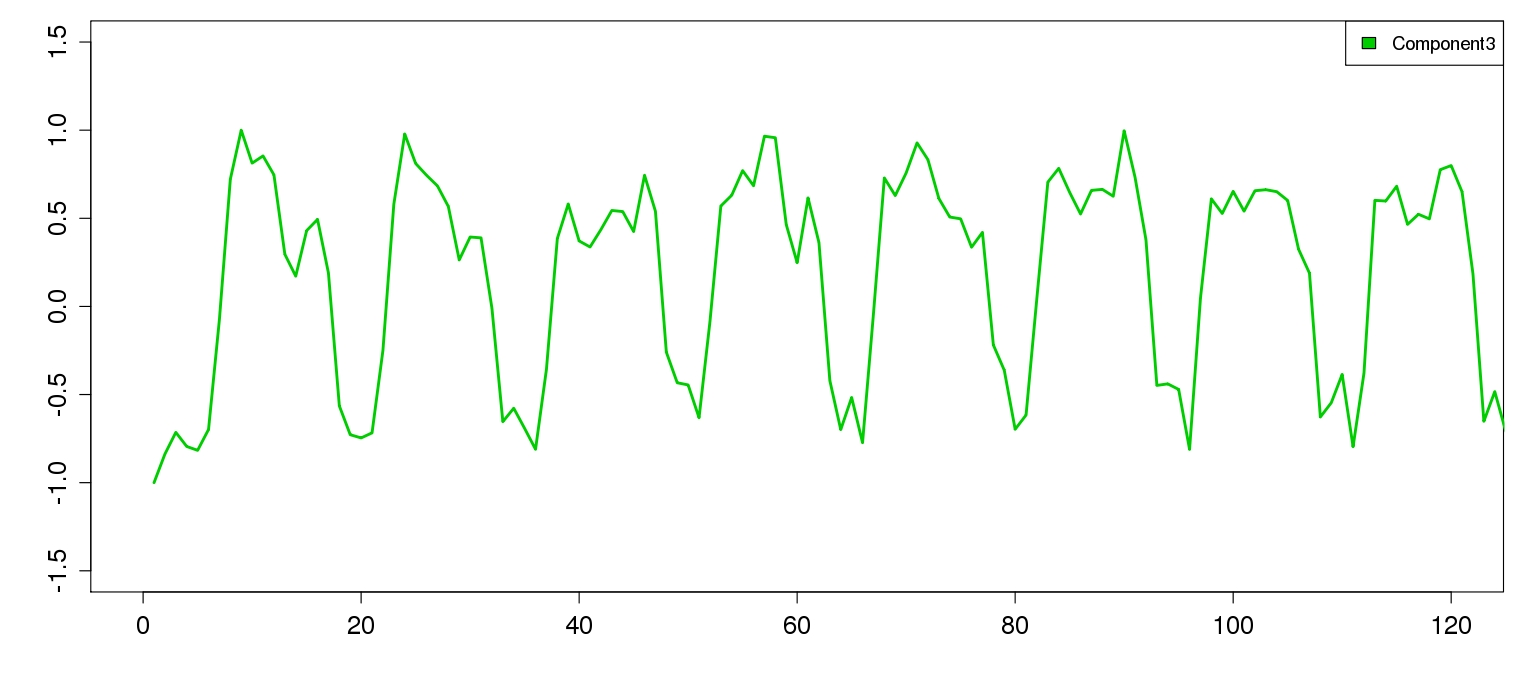}\\
\end{tabular}
\caption{Color center mapping. Left: Visual stimulation alternated blocks presenting chromatic and achromatic version of Mondrian like patterns. Middle : Component at the frequency of the stimulation of the visual system extracted using temporal ICA. Right: Component at the frequency of the stimulation of the visual system extracted using spatial ICA.}
\label{mond_results}
\end{figure}
As shown in Figure \ref{overlays}, the voxels containing the components extracted using temporal and spatial ICA were localized as expected in the same ventral cortical region, called V4-V8, known to be sensitive to color perception (\cite{Wade02}).

\begin{figure}[h!]
\centering
\includegraphics[width=0.7\textwidth ]{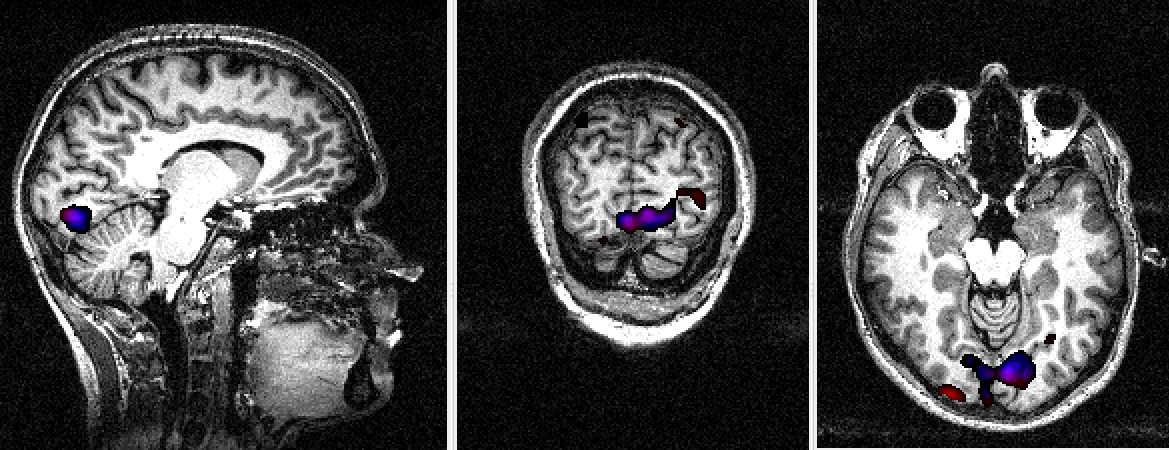}
\caption{Localization of the most active voxels for the extracted component at the stimulus frequency using temporal (red) and spatial ICA (blue). Anatomical view: Left: sagittal; Middle: Coronal; Right: Transverse. There is a strong overlap between the voxels corresponding to temporal and spatial analysis. As expected, they are positioned in a cortical region known to be color sensitive.}
\label{overlays}
\end{figure}

\section{Discussion}
In addition to the standard model-driven approach (GLM),  where the time course of the stimulus pattern is convolved with a hemodynamic response function and 
used as a predictor to detect brain activation, the analysis of fMRI signals could benefit from data-driven approaches such as ICA. 
ICA seems a powerful method to reveal brain activation patterns with a good temporally and spatially accuracy or to extract noise components from the data  (\cite{McKeown06}).
The strength of ICA is its ability to reveal hidden spatio-temporal structure without the definition of a specified \textit{a priori} model. Since its first application to fMRI data analysis (\cite{McKeown95}), ICA have been used in various brain function studies. For example, ICA was successfully applied to investigate the cortical networks related to natural multimodal stimulation (\cite{Malinen07}) or natural viewing conditions (\cite{Bartels04}); situations in which activity is present in various brain sites and no \textit{a priori} knowledge about the spatial location or about the activity waveforms were available. In (\cite{Bartels04}), sICA allowed to segregate a multitude of functionally specialized cortical and subcortical regions because they exhibit specific differences in the activity time course of the voxels belonging to them. In (\cite{Seifritz02}), tICA revealed un-predicted and un-modeled responses in the auditory system.
Following (\cite{Calhoun01b} or \cite{Malinen07}) GLM-derived activations are spatially less extensive and comprised only sub-areas of the ICA detected activations. ICA can both detect responses that are consistently and transiently task-related while GLM is restricted to the former (\cite{McKeown95}, \cite{Hu05}). \\

A number of ICA approaches have been proposed for fMRI data analysis. A comparison of some algorithms for fMRI analysis can be found in  (\cite{Correa07}). There are two largely used Matlab toolboxes, GIFT (\url{http://www.nitrc.org/projects/gift/}) implementing the FastICA algorithm (\cite{Hyvarinen99}), which maximizes the non-gaussianity of estimated sources and JADE, which relies on a joint approximate diagonalization of eigenmatrices (\cite{Cardoso93} ). Probabilistic ICA (PICA) is embedded in the FSL package (\cite{Beckmann04}), a library of tools for neuroimaging data analysis (\url{http://www.fmrib.ox.ac.uk/fsl/}). 
In this paper, we propose a new version of the \R package \pkg{AnalyzeFMRI}, dedicated to the fMRI data analysis, for temporal and spatial IC analysis. We reused, with some memory improvements, the implementation of the FastICA algorithm  proposed in the \R package \pkg{fastICA}. Essentially for tractability considerations, spatial ICA is generally used in the context of neuroimaging. However, the temporal independence of sources can be supposed in some applications. In this case, only a small part of the brain is considered (\cite{Calhoun01b}; \cite{Seifritz02}). We have shown using a classical linear algebra result that temporal ICA can be tractable on large fMRI data sets. Based on simulated data and real functional data sets, we have demonstrated the applicability of the package proposed for spatial and temporal ICA. As we have seen with the traveling wave case, sinusoids with the same frequency but presenting different phases are not independent and then can not be extracted using ICA. A possible solution to this problem would be to use a least square approach by imposing a strong \textit{a priori} on the sources: the $i^{th}$ source is $S_i(t)=\sin(2\pi fU_t+\phi_i)$ where the frequency $f$ is supposed to be known.
We then search estimated sources $Y_1,...,Y_m$ under this specific form that can be written as a linear combination of the observed signals: $Y_i=a_{i1}X_1(t)+a_{i2}X_2(t)+\ldots+a_{im}X_m(t), t=1, \ldots,n$. The least square problem to optimize (numerically) is then
$$\sum_{t=1}^n \left(\sin(2\pi fU_t+\phi_i)-a_{i1}X_1(t)+a_{i2}X_2(t)+\ldots+a_{im}X_m(t)\right)^2, i=1, \ldots,m.$$
Note that we could also differentiate with respect to the $\phi_i$'s and the $a_{ij}$'s to simplify the computation.

Several extensions should be inserted in the future.  The package should be extended for dealing with group studies. Indeed, ICA generates a large number of components for each subject and obviously larger for a cohort of subjects.
Several methods have been proposed for dealing specifically with group studies (\cite{Svensen02}; \cite{Esposito05}; \cite{Varoquaux10})  and to facilitate the identification of components which are spurious or reproducible (\cite{Himberg04}; \cite{Cordes07}; \cite{Vigario08}; \cite{Wang08}). The sorting of relevant components can be performed using several indexes such as correlation coefficient with a reference function (\cite{Hu05}) or power spectrum (\cite{Moritz03}). A possible improvement would be to use a more general measure of dependence as the one provided in \cite{Beran07}.

To resume, ICA is a powerful data-driven technique that allows neuroscientists to explore the intrinsic structure of data and to alleviate the need for explicit \textit{a priori} about the neural responses. We propose with the TS-ICA extension to the \R package \pkg{AnalyzeFMRI} a robust tool for the application both of spatial and temporal ICA to fMRI data.

\section*{Acknowledgments}
C\'ecile Bordier is recipient of a grant from Institut National de la Sant\'e et de la Recherche Scientifique (INSERM). 
Pierre Lafaye de Micheaux is recipient of a grant from the Natural Sciences and Engineering Research Council (NSERC) of Canada. The authors would also like to thank Professor Christian Jutten for many helpful comments.

\bibliography{biblio-jss}

\end{document}